\documentclass[10pt,a4paper]{amsart}
\usepackage[latin1]{inputenc}
\usepackage{mathrsfs}
\usepackage{dsfont}
\usepackage{graphicx}
\usepackage[]{hyperref}
\usepackage{esint}
\newtheorem{theorem}{Theorem}
\newtheorem{proposition}{Proposition}
\newtheorem{remark}{Remark}
\newtheorem{lemma}{Lemma}

\newcommand{\ii}{\infty}

\newcommand\R{{\ensuremath {\mathbb R} }}
\newcommand\C{{\ensuremath {\mathbb C} }}
\newcommand\N{{\ensuremath {\mathbb N} }}

\newcommand\cS{\mathcal{S}}
\renewcommand\phi{\varphi}

\newcommand{\gS}{\mathfrak{S}}

\newcommand{\cP}{\mathcal{P}}
\newcommand{\cQ}{\mathcal{Q}}
\newcommand{\cC}{\mathcal{C}}
\newcommand{\cR}{\mathcal{R}}

\newcommand{\cB}{\mathcal{B}}
\newcommand{\cK}{\mathcal{K}}

\newcommand{\cF}{\mathcal{F}}
\newcommand{\cN}{\mathcal{N}}

\newcommand{\cL}{\mathcal{L}}

\newcommand{\curv}{\mathscr{C}}
\renewcommand{\epsilon}{\varepsilon}
\newcommand\pscal[1]{{\ensuremath{\left\langle #1 \right\rangle}}}

\def\tr{{\rm Tr}}

\def\div{{\rm div \;}}

\def\NN{\N}

\def\RR{\R}
\def\CC{\C}

\title[Dielectric Permittivity of Crystals in the rHF approximation]{The
  Dielectric Permittivity of Crystals in the reduced Hartree-Fock
  approximation} 

\author[\'E. Cancès]{\'Eric CANC\`ES}
 \address{Universit\'e Paris-Est, CERMICS, Project-team Micmac, INRIA-Ecole des Ponts,  6 \& 8 avenue Blaise Pascal, 77455 Marne-la-Vall\'ee Cedex 2, France.}
  \email{cances@cermics.enpc.fr} 
 
\author[M. Lewin]{Mathieu LEWIN}
 \address{CNRS \& Laboratoire de Math\'ematiques UMR 8088, Universit\'e de Cergy-Pontoise, 2 Avenue Adolphe Chauvin, 95302 Cergy-Pontoise Cedex, France.}
  \email{Mathieu.Lewin@math.cnrs.fr}

\date{March 11, 2009}

\begin{document}
\begin{abstract}
In a recent article (Cancès, Deleurence and Lewin, \textit{Commun. Math. Phys.} \textbf{281} (2008), pp. 129--177), we have rigorously derived, by
means of bulk limit arguments, a new variational model to describe the
electronic ground state of insulating or semiconducting crystals in the
presence of local defects. In this so-called reduced Hartree-Fock model, the ground state electronic density
matrix is decomposed as $\gamma = 
\gamma^0_{\rm per} + Q_{\nu,\epsilon_{\rm F}}$, where $\gamma^0_{\rm
  per}$ is the ground state density matrix of the host crystal and
$Q_{\nu,\epsilon_{\rm F}}$ the modification of the electronic density matrix
generated by a modification $\nu$ of the nuclear charge of the host
crystal, the Fermi level $\epsilon_{\rm F}$ being kept fixed. The
purpose of the present article is twofold. First, we study more in
details the mathematical properties of the density matrix
$Q_{\nu,\epsilon_{\rm F}}$ (which is known to
be a self-adjoint Hilbert-Schmidt operator on $L^2(\R^3)$). We show in
particular that if $\int_{\RR^3} \nu \neq 0$,
$Q_{\nu,\epsilon_{\rm F}}$ is not trace-class. Moreover, the associated
density of charge is not in $L^1(\R^3)$ if the crystal exhibits
anisotropic dielectric properties. These results are obtained by
analyzing, for a small defect $\nu$, the linear and nonlinear terms of
the resolvent expansion of $Q_{\nu,\epsilon_{\rm F}}$. 
Second, we show that, after an appropriate rescaling, the potential
generated by the microscopic total charge (nuclear plus electronic
contributions) of the crystal in the presence of the defect, converges
to a homogenized electrostatic potential solution to a Poisson equation
involving the macroscopic dielectric permittivity of the crystal. This
provides an alternative (and rigorous) derivation of the
Adler-Wiser formula. 
\end{abstract}

\maketitle


\setcounter{tocdepth}{1}

\section{Introduction}

The electronic structure of crystals with local defects has been the
topic of a huge number of articles and monographs in the Physics
literature. On the other hand, the mathematical foundations of
the corresponding models still are largely unexplored.

In \cite{CanDelLew-08a}, we have introduced a variational framework
allowing for a rigorous characterization of the electronic ground state of
insulating or semi-conducting crystals with local
defects, within the reduced Hartree-Fock
setting. Recall that the reduced Hartree-Fock (rHF) model -- also called Hartree model --
\cite{Solovej-91} is a nonlinear approximation of the $N$-body Schrödinger theory, where the state of the electrons is described by a \emph{density matrix} $\gamma$, i.e. a self-adjoint operator acting on $L^2(\R^3)$ (in order to simplify the
notation, the spin variable will be omitted in the whole paper) such that $0\leq\gamma\leq 1$ and whose trace equals the total number of electrons in the system. The rHF model may be obtained from the usual Hartree-Fock model \cite{LieSim-77} by neglecting the so-called \emph{exchange term}. It may also be obtained from the extended
Kohn-Sham model \cite{DreGro-90} by setting to zero the
exchange-correlation functional. 

Our variational model is derived from the
supercell approach commonly used in numerical simulations, by letting the
size of the supercell go to infinity. It is found
\cite{CatBriLio-01,CanDelLew-08a} that the density matrix of the perfect
crystal converges in the limit to a periodic density matrix
$\gamma^0_{\rm per}$ describing the infinitely many electrons of the
periodic crystal (Fermi sea). In the presence of a defect modelled by a
nuclear density of charge $\nu$, the density matrix of the electrons converges in the limit to a state $\gamma$ which can be decomposed as \cite{CanDelLew-08a}
\begin{equation}
 \gamma = \gamma^0_{\rm per} + Q_{\nu,\epsilon_{\rm F}}
\label{eq:SCF_intro}
\end{equation}
where $Q_{\nu,\epsilon_{\rm F}}$ accounts
for the modification of the electronic density matrix induced by a
modification $\nu$ of the nuclear charge of the crystal. The operator 
$Q_{\nu,\epsilon_{\rm F}}$ depends on $\nu$ as well as on the Fermi
level $\epsilon_{\rm F}$, which controls the total charge of the defect. 

Loosely speaking, the operator
$Q_{\nu,\epsilon_{\rm F}}$ describing the modification of the
electronic density matrix should be small when $\nu$ itself is 
small. But it is \emph{a priori} not clear for which norm this really
makes sense. 
The mathematical difficulties of such a model lay in the fact that the
one-body density matrices $\gamma_{\rm per}^0$ and $\gamma$ are infinite
rank operators (usually orthonormal projectors), representing
Hartree-Fock states with infinitely many interacting electrons. 

Following previous results on a QED model
\cite{HaiLewSer-05a,HaiLewSerSol-07,GraLewSer-08} describing
relativistic electrons interacting with the self-consistent Dirac sea,
we have characterized in \cite{CanDelLew-08a} the solution $Q_{\nu,\epsilon_{\rm F}}$ of \eqref{eq:SCF_intro} as
the minimizer of a certain energy functional $E_{\nu,\epsilon_{\rm F}}$
on a convex set $\cK$ that will be defined later. This
procedure leads to the information that  
\begin{equation}
Q_{\nu,\epsilon_{\rm F}}\in \gS_2,\qquad Q^{--}_{\nu,\epsilon_{\rm F}},\
Q^{++}_{\nu,\epsilon_{\rm F}} \in\gS_1. 
\label{info_min} 
\end{equation}
In the whole paper we use the notation 
$$
\begin{array}{ll}
Q^{--} := \gamma^0_{\rm per} Q \gamma^0_{\rm per} & \qquad 
Q^{-+} := \gamma^0_{\rm per} Q (1-\gamma^0_{\rm per}) \\
Q^{+-} := (1-\gamma^0_{\rm per}) Q \gamma^0_{\rm per} & \qquad 
Q^{++} := (1-\gamma^0_{\rm per}) Q (1-\gamma^0_{\rm per}),
\end{array}
$$
and we denote by $\gS_1$ and $\gS_2$ respectively the spaces of trace-class
and Hilbert-Schmidt operators on $L^2(\R^3)$. A definition of these spaces is
recalled at the beginning of Section~\ref{sec:RHFperturbed} for the
reader's convenience.  

Property \eqref{info_min} implies that the operator
$Q_{\nu,\epsilon_{\rm F}}$ is compact, but it does
\emph{not} mean \emph{a priori} that it is trace-class. This
mathematical difficulty complicates the definition of the density of
charge. For $Q \in
\gS_1$, the density of charge can be defined by 
$\rho_{Q}(x)=Q(x,x)$ where $Q(x,x')$ is the integral kernel of $Q$;
it satisfies $\int_{\R^3}\rho_{Q}=\tr(Q)$. Let us emphasize that these
formulae only make sense when $Q$ is trace-class.   

In \cite{CanDelLew-08a} we have been able to define the density
$\rho_{Q_{\nu,\epsilon_{\rm F}}}$ associated with $Q_{\nu,\epsilon_{\rm
    F}}$ by a duality argument, but this only gave us the following information:
\begin{equation}
\rho_{Q_{\nu,\epsilon_{\rm F}}}\in L^2(\R^3) \quad \text{and} \quad
\int_{\R^3} \int_{\R^3} \frac{\rho_{Q_{\nu,\epsilon_{\rm F}}}(x)\rho_{Q_{\nu,\epsilon_{\rm F}}}(x')}{|x-x'|}dx\,dx'<\ii.
\label{info_rho} 
\end{equation}
Also, following \cite{HaiLewSer-05a} and using \eqref{info_min} one can
define the electronic charge of the state counted relatively to the Fermi sea $\gamma^0_{\rm per}$ via
$$
\tr_0(Q_{\nu,\epsilon_{\rm F}}):=\tr(Q^{++}_{\nu,\epsilon_{\rm
    F}})+\tr(Q^{--}_{\nu,\epsilon_{\rm F}}).
$$
It can be shown that when $\nu$ is small enough (in an appropriate sense
precised below), 
$$\tr_0(Q_{\nu,\epsilon_{\rm F}})=0,$$
hence the Fermi sea stays overall neutral in the presence of a small defect.

In \cite{CanDelLew-08a}, we
left open two very natural questions: 
\begin{enumerate}
\item is $Q_{\nu,\epsilon_{\rm F}}$ trace-class?
\item if not, is $\rho_{Q_{\nu,\epsilon_{\rm F}}}$ nevertheless an
  integrable function? 
\end{enumerate}
The purpose in the present article is twofold. First, we prove that
$Q_{\nu,\epsilon_{\rm F}}$ is {\em never} trace-class when
$\int_{\R^3}\nu\neq0$, and that, {\em in general}, 
$\rho_{Q_{\nu,\epsilon_{\rm F}}}$ is {\em not} an integrable function
  (at least for anisotropic dielectric crystals). These unusual mathematical
  properties are in fact directly related to the dielectric properties
  of the host crystal. They show in particular that the
approach of \cite{CanDelLew-08a} involving the complicated variational
set $\cK$ cannot {\it a priori} be simplified by replacing $\cK$ with a
simpler variational set (a subset of $\gS_1$ for instance).

In a second part, we show that our variational model allows to recover the
Adler-Wiser formula \cite{Adler-62,Wiser-63} for the electronic contribution
to the macroscopic dielectric constant of the perfect crystal, by means of a homogenization
argument. More precisely, we rescale a fixed density $\nu \in
L^1(\R^3) \cap L^2(\R^3)$ as follows
$$
\nu_\eta(x):=\eta^3\nu(\eta x),
$$
meaning that we submit the Fermi sea to a modification of the external
potential which is very spread out in space. We consider the
(appropriately rescaled) total electrostatic potential 
$$
W_\nu^\eta(x):=\eta^{-1}\left[\left(\nu_\eta-\rho_{Q_{\nu_\eta,\epsilon_{\rm
          F}}}\right)\star|\cdot|^{-1}\right]\left(\eta^{-1} x\right)
$$
of the nonlinear system consisting of the density $\nu_\eta$
and the self-consistent variation $\rho_{Q_{\nu_\eta,\epsilon_{\rm F}}}$
of the density of the Fermi sea. We prove that $W_\nu^\eta$ converges 
weakly to $W_\nu$, the unique solution in ${\mathcal
  S}'(\R^3)$ of the elliptic equation 
\begin{equation}
\boxed{-\div(\epsilon_{\rm M}\nabla W_\nu)=4\pi\nu}
\label{eq:intro_eq} 
\end{equation}
where $\epsilon_{\rm M}$ is the so-called \emph{macroscopic dielectric
  permittivity}\footnote{To be precise, it is only the \emph{electronic
    part} of the macroscopic dielectric permittivity, as we do not take
  into account here the contribution originating from the relaxation of
  the nuclei of the lattice (the nuclei are fixed in our approach).}, a
$3\times3$ 
symmetric, coercive, matrix which only depends on the perfect crystal,
and can be computed from the Bloch-Floquet decomposition of the mean-field
Hamiltonian. As we will explain in details, the occurence of the dielectric permittivity $\epsilon_{\rm M}$, or more precisely the fact that in general $\epsilon_{\rm M}\neq 1$, is indeed related to the properties that $Q_{\nu,\epsilon_{\rm F}}$ is not trace-class and $\rho_{Q_{\nu,\epsilon_{\rm F}}}$ is not in $L^1(\R^3)$.

This article is organized as follows. In Section~\ref{sec:rHF}, we briefly
present the reduced Hartree-Fock model for molecular systems with finite
number of electrons, for perfect crystals and for
crystals with local defects. In Section~\ref{sec:response_eff} we study
the linear response of the perfect crystal to a variation of the
effective potential, the nonlinear response being the matter of
Section~\ref{sec:expansion_Q_V}.
Note that the results contained in this section can be applied to the linear model (non-interaction electrons), to the reduced
Hartree-Fock model, as well as to the Kohn-Sham LDA model.
We then focus in Section~\ref{sec:response_ext} on the 
response of the reduced Hartree-Fock ground state of the crystal to a small
modification of the external potential generated by a modification
$\nu$ of the nuclear charge. We prove that for $\nu$ small
enough and such that $\int_{\R^3} \nu \neq 0$, one has 
$\tr_0(Q_{\nu,\epsilon_{\rm F}}) = 0$ while the Fourier transform
$\widehat{\rho}_{Q_{\nu,\epsilon_{\rm F}}}(k)$ of
${\rho}_{Q_{\nu,\epsilon_{\rm F}}}$ does not converge to $0$ when
$k$ goes to zero, yielding $Q_{\nu,\epsilon_{\rm F}} \notin \gS_1$. We
also prove that if the host crystal exhibits anisotropic dielectric
properties, $\widehat{\rho}_{Q_{\nu,\epsilon_{\rm F}}}
(k)$ does not have a limit
at $k =0$, which implies that $\rho_{Q_{\nu,\epsilon_{\rm F}}} \notin
L^1(\R^3)$. Finally, it is shown in Section~\ref{sec:macroscopic} that,
after rescaling, the potential generated by the microscopic total charge
(nuclear plus electronic contributions) of
the crystal in the presence of the defect, converges to a homogenized
electrostatic potential solution to the Poisson equation \eqref{eq:intro_eq} involving the
macroscopic dielectric permittivity of the crystal.
All the proofs are gathered in Section~\ref{sec:proofs}.

\section{The reduced Hartree-Fock model for molecules and crystals}
\label{sec:rHF}

In this section, we briefly recall the reduced Hartree-Fock model for finite systems, perfect crystals and crystals with a localized defect.

\subsection{Finite system}
Let us first consider a molecular system
containing $\cN$ non-relativistic quantum electrons and a set of nuclei
having a 
density of charge $\rho^{\rm nuc}$. If for instance the system contains $M$
nuclei of charges $z_1,\cdots,z_K \in \NN \setminus \left\{0\right\}$
located at $R_1,\cdots,R_K \in\R^3$, then  
$$
\rho^{\rm nuc}(x):=\sum_{k=1}^K z_k \, m_k(x-R_k),
$$
where $m_1,\cdots,m_K$ are probability measures on $\R^3$. Point-like
nuclei correspond to $m_k=\delta$ (the Dirac 
measure) while smeared nuclei are modeled by smooth, nonnegative,
radial, compactly supported 
functions $m_k$ such that $\int_{\RR^3} m_k = 1$. 

The electronic energy of the system of $\cN$ electrons in the reduced Hartree-Fock model reads
\cite{Solovej-91,CanDelLew-08a} 
\begin{equation}
\mathcal{ E}_{\rho^{\rm nuc}}(\gamma) = \tr \left( -
  \frac 1 2 \Delta \gamma \right) - \int_{\RR^3} \rho_\gamma \left(
  \rho^{\rm nuc} \star |\cdot|^{-1}\right) 
+ \frac 1 2 D \left( \rho_\gamma , \rho_\gamma \right).
\label{def_rHF}
\end{equation}
The above energy is written in atomic units, i.e. $\hbar=1$, $m=1$,
$e=1$ and $\frac
1{4\pi\epsilon_0}=1$ where $m$ is the mass of the electron, $e$ the
elementary charge, $\hbar$ the reduced Planck constant and
 $\epsilon_0$ the dielectric permittivity of the vacuum.
The first term in the right-hand side of \eqref{def_rHF} is the kinetic
energy of the electrons 
and $D(\cdot,\cdot)$ is the classical Coulomb interaction, which
reads for $f$ and $g$ in $L^{\frac 65}(\R^3)$ as
\begin{equation}
D(f,g) = \int_{\R^3} \int_{\R^3} \frac{f(x) \, g(y)}{|x-y|}dx \, dy=4\pi\int_{\R^3}\frac{\overline{\widehat{f}(k)}\widehat{g}(k)}{|k|^2}dk,
\label{def_D_f_g} 
\end{equation}
where $\widehat{f}$ denotes the Fourier transform of $f$. Here and in the
sequel, we use the normalization convention consisting in defining
$\widehat{f}(k)$ as
$$
\widehat{f}(k) = (2\pi)^{-\frac 32} \int_{\R^3} f(x) e^{-ik\cdot x} \, dx.
$$
In this
mean-field model, the state of the $\cN$ electrons is described by the
one-body density matrix $\gamma$, which is an element of the following
class 
\begin{equation*}
\mathcal{ P}^{\cN} = \bigg\{ \gamma \in \cS(L^2(\R^3)) \; | \; \ 0 \le \gamma
  \le 1, \ \tr(\gamma) = \cN, \ \tr\left(\sqrt{-\Delta} \gamma\sqrt{-\Delta}\right) < \infty \bigg\},
\end{equation*}
$\cS(L^2(\R^3))$ denoting the space of bounded self-adjoint operators on $L^2(\R^3)$. Also we define $\tr(-\Delta\gamma):=\tr(\sqrt{-\Delta}\gamma\sqrt{-\Delta})$ which makes sense when $\gamma\in\cP^\cN$.
The set $\cP^\cN$ is the closed convex hull of the set of orthogonal projectors of rank $\cN$ acting on $L^2(\R^3)$ and having a finite kinetic energy.

The function $\rho_\gamma$ appearing in \eqref{def_rHF} is the density
associated with the operator~$\gamma$, defined by $\rho_\gamma(x) =
\gamma(x,x)$ where $\gamma(x,y)$ is the kernel of the trace class
operator  
$\gamma$. Notice that for all $\gamma \in \mathcal{P}^\cN$, one has
$\rho_\gamma\geq0$ and $\sqrt{\rho_\gamma}\in H^1(\R^3)$, hence the
last two terms of \eqref{def_rHF} are well-defined, since $\rho_\gamma\in
L^1(\R^3) \cap L^3(\R^3)\subset L^{\frac 65}(\R^3)$.   

It can be proved (see the appendix of \cite{Solovej-91}) that if $\cN \le
\sum_{k=1}^M z_k $ (neutral or positively charged 
systems), the variational problem 
\begin{equation}
\inf \left\{ \mathcal{ E}_{\rho^{\rm nuc}}(\gamma),
  \; \gamma \in \mathcal{ P}^\cN \right\} 
\label{def_min_rHF}
\end{equation}
has a minimizer $\gamma$ and that all the minimizers share the same
density $\rho_\gamma$.

\subsection{The perfect crystal}
The above model describes a \emph{finite} system of $\cN$ electrons in
the electrostatic field created by the density $\rho^{\rm
  nuc}$. Our goal is to describe an \emph{infinite} crystalline material
obtained in the bulk limit $\cN\to\ii$. In fact we shall
consider two such systems. The first one is the periodic crystal obtained when, in the bulk limit, the nuclear density approaches the periodic nuclear distribution of the perfect
crystal:
\begin{equation}
 \rho^{\rm nuc}\rightarrow \rho_{\rm per}^{\rm nuc},
\label{case_periodic}
\end{equation}
$\rho^{\rm nuc}_{\rm per}$ being a $\cR$-periodic distribution, $\cR$
denoting a periodic lattice of $\R^3$. The second system
is the previous crystal in presence of a local defect: 
\begin{equation}
\rho^{\rm nuc}\rightarrow \rho_{\rm per}^{\rm nuc}+\nu. 
\label{case_defect}
\end{equation}

The density matrix $\gamma^0_{\rm per}$ of the \emph{perfect
  crystal} obtained in the bulk limit \eqref{case_periodic} is
unique~\cite{CanDelLew-08a}.   
It is the unique solution to the self-consistent equation
\begin{equation}
\boxed{\gamma^0_{\rm per}=1_{(-\ii;\epsilon_{\rm F}]}(H^0_{\rm per})}
\label{eq:SCF0}
\end{equation}
\begin{equation}
\boxed{H^0_{\rm per} = - \frac 1 2 \Delta + V_{\rm per}}
\label{def_H_0}
\end{equation}
where $V_{\rm per}$ is a $\cR$-periodic function satisfying
$$
-\Delta V_{\rm per} = 4 \pi \left( \rho_{\rm per}^0 - \rho_{\rm
    per}^{\rm nuc} \right),\quad\text{ with }\quad \rho_{\rm per}^0(x) =
\gamma_{\rm  per}^0(x,x),$$
and where $\epsilon_{\rm F} \in \RR$ is the Fermi level. The potential
$V_{\rm per}$ is defined up to an additive constant; if $V_{\rm
  per}$ is replaced with $V_{\rm per}+C$, $\epsilon_{\rm F}$ has to be
replaced with $\epsilon_{\rm F}+C$, in such a way that
$\gamma^0_{\rm per}$ remains unchanged. 
The function $V_{\rm per}$ being in $L^2_{\rm per}(\R^3)$, it defines a $\Delta$-bounded operator on $L^2(\R^3)$ with relative bound
  zero (see \cite[Thm XIII.96]{ReeSim4}) and therefore $H^0_{\rm  per}$
  is self-adjoint on $L^2(\R^3)$ with domain $H^2(\R^3)$. Besides, the
  spectrum of $H^0_{\rm per}$ is purely 
  absolutely continuous, composed of bands as stated in \cite[Thm
  1-2]{Thomas-73} and \cite[Thm XIII.100]{ReeSim4}.

More precisely, denoting by $\cR^\ast$ the reciprocal lattice, by $\Gamma$
the unit cell, and by $\Gamma^\ast$ the Brillouin zone,
 we have 
$$
\sigma(H^0_{\rm per})=\bigcup_{n\geq1,\
  q\in\Gamma^\ast} \left\{\epsilon_{n,q}\right\}
$$
where for all $q \in \Gamma^\ast$, $(\epsilon_{n,q})_{n\geq1}$ is the
non-decreasing sequence formed by the eigenvalues (counted with their
multiplicities) of the operator 
$$
(H^0_{\rm per})_q = 
-\frac{1}{2} \Delta-iq\cdot\nabla +\frac{|q|^2}{2} +V_{\rm per}
$$ 
acting on 
$$
L^2_{\rm per}(\Gamma):= \left\{ u \in L^2_{\rm loc}(\R^3) \; | \; u\ 
  \mbox{$\cR$-periodic} \right\},
$$
endowed with the inner product
$$
\langle u,v \rangle_{L^2_{\rm per}} = \int_\Gamma \overline{u} \, v.
$$
We denote by $(u_{n,q})_{n \ge 1}$ an
orthonormal basis of $L^2_{\rm per}(\Gamma)$ consisting of associated
eigenfunctions. The spectral decomposition of $(H^0_{\rm
  per})_q$ thus reads
\begin{equation} \label{eq:dec_HOperq}
(H^0_{\rm per})_q = \sum_{n=1}^\infty \epsilon_{n,q} |u_{n,q}\rangle
\langle u_{n,q}|.
\end{equation}
Recall that according to the Bloch-Floquet theory \cite{ReeSim4}, any
function $f \in L^2(\R^3)$ can be decomposed as 
$$
f(x) = \fint_{\Gamma^\ast} f_q(x) \, e^{iq\cdot x} dq,
$$
where $\fint_{\Gamma^\ast}$ is a notation for
$|\Gamma^\ast|^{-1}\int_{\Gamma^\ast}$ and where the functions $f_q$ are
defined by 
\begin{equation}
f_q(x)=\sum_{R\in\cR}f(x+R) e^{-iq\cdot
  (x+R)}=\frac{(2\pi)^{\frac 32}}{|\Gamma|}
\sum_{K\in\cR^\ast}\widehat{f}(q+K)e^{iK\cdot x}.
\label{def_Bloch_Floquet}
\end{equation}
For almost all $q \in \R^3$, $f_q \in L^2_{\rm per}(\Gamma)$. Besides, 
$f_{q+K}(x)=f_q(x)e^{-iK\cdot x}$ for all $K \in \cR^\ast$ and almost
all $q \in \R^3$. Lastly,
$$
\|f\|_{L^2(\RR^3)}^2 = \fint_{\Gamma^\ast} \|f_q\|_{L^2_{\rm per}(\Gamma)}^2 \, dq.
$$

If the crystal possesses $N$ electrons per unit cell, the
Fermi level $\epsilon_{\rm F}$ is chosen to ensure the correct
charge per unit cell: 
\begin{equation}
N=\sum_{n\geq1} \left|\{q\in\Gamma^\ast\ |\
  \epsilon_{n,q}\leq{\epsilon_{\rm F}}\}\right|.
\label{def_Z}
\end{equation}
In the rest of the paper we will assume that the system is an insulator
(or a semi-conductor) in the sense that the $N^{\rm th}$ band is
strictly below the $(N+1)^{\rm st}$ band:
$$
\Sigma_N^+:=\max_{q\in\Gamma^\ast}\epsilon_{N,q}<\min_{q\in\Gamma^\ast}\epsilon_{N+1,q}:=\Sigma_{N+1}^-. 
$$
In this case, one can choose for $\epsilon_{\rm F}$ any number in the range
$(\Sigma_N^+,\Sigma_{N+1}^-)$. For simplicity we will take in the following 
$$\epsilon_{\rm F}=\frac{\Sigma_N^++\Sigma_{N+1}^-}{2}
$$
and denote by 
$$
g = \Sigma_{N+1}^- - \Sigma_N^+ > 0
$$
the band gap.

\subsection{The perturbed crystal}
\label{sec:RHFperturbed}
Before turning to the model for the crystal with a defect which was introduced in \cite{CanDelLew-08a},  let us recall that a bounded linear operator $Q$ on $L^2(\R^3)$ is said to be
\emph{trace-class} \cite{ReeSim4,Simon-79} if
$\sum_i\pscal{\phi_i,\sqrt{Q^\ast Q}\phi_i}_{L^2}<\ii$ for
some orthonormal basis $(\phi_i)$ of $L^2(\R^3)$. Then
$\tr(Q)=\sum_i\pscal{\phi_i,Q\phi_i}_{L^2}$ is well-defined and does not
depend on the chosen basis. If $Q$ is not trace-class, it may happen
that the series $\sum_i\pscal{\phi_i,Q\phi_i}_{L^2}$ converges for one
specific basis but not for another one. This will be the case for our
operators $Q_{\nu,\epsilon_{\rm F}}$.  

A compact operator $Q=\sum_i\lambda_i|\phi_i\rangle\langle\phi_i| \in
{\mathcal S}(L^2(\R^3))$ is trace-class when its eigenvalues are
summable, $\sum_i|\lambda_i|<\ii$. Then the density 
$$
\rho_Q(x) = Q(x,x) = \sum_{i=1}^{+\infty} \lambda_i |\phi_i(x)|^2
$$
is a function of $L^1(\R^3)$ and 
$$
\tr(Q) = \sum_{i=1}^{+\infty} \lambda_i = \int_{\R^3} \rho_Q.
$$
On the other hand, a Hilbert-Schmidt operator $Q$ is by definition such
that $Q^\ast Q$ is trace-class. 

\medskip

We now describe the results of \cite{CanDelLew-08a} dealing with the perturbed crystal. We have proved in \cite{CanDelLew-08a} by means of bulk limit arguments
that the 
ground state density matrix of the crystal with nuclear charge density
$\rho^{\rm nuc}_{\rm per} + \nu$ reads
$$
\gamma = \gamma^0_{\rm per} + Q_{\nu,\epsilon_{\rm F}}
$$
where $Q_{\nu,\epsilon_{\rm F}}$ is obtained by minimizing the following energy
functional
\begin{equation}
E_{\nu,\epsilon_{\rm F}}(Q) =  \tr(|H^0_{\rm per}-\epsilon_{\rm F}|(Q^{++}-Q^{--})) -
\int_{\RR^3} \rho_Q (\nu \star |\cdot|^{-1}) + \frac 1 2 D(\rho_Q,\rho_Q)
\end{equation}
on the convex set 
\begin{equation} \label{eq:defK}
\cK = \big\{ Q \in \cQ \; | \; -\gamma^0_{\rm per} \le Q \le 1 - 
\gamma^0_{\rm per} \big\}
\end{equation}
where
\begin{eqnarray}
\cQ &=& \big\{ Q \in \gS_2 \; |  \; Q^\ast = Q, \; \; Q^{--} \in \gS_1, \; Q^{++} \in \gS_1, \label{eq:defQ}
 \\ && \qquad
\qquad \quad |\nabla|Q \in \gS_2, \; |\nabla|Q^{--}|\nabla| \in \gS_1,
\; |\nabla|Q^{++}|\nabla| \in \gS_1  \big\}. \nonumber 
\end{eqnarray}
Recall that  $\gS_1$ and $\gS_2$ respectively denote the spaces of trace-class
and Hilbert-Schmidt operators on $L^2(\R^3)$ and that
$$
\begin{array}{ll}
Q^{--} := \gamma^0_{\rm per} Q \gamma^0_{\rm per} & \qquad 
Q^{-+} := \gamma^0_{\rm per} Q (1-\gamma^0_{\rm per}) \\
Q^{+-} := (1-\gamma^0_{\rm per}) Q \gamma^0_{\rm per} & \qquad 
Q^{++} := (1-\gamma^0_{\rm per}) Q (1-\gamma^0_{\rm per}).
\end{array}
$$
It is proved in \cite{CanDelLew-08a} that although a generic operator $Q
\in \cQ$ is 
not trace-class, it can be associated a generalized trace $\tr_0(Q) =
\tr(Q^{++})+ \tr(Q^{--})$ and
a density $\rho_Q \in L^2(\R^3)\cap \cC$ where the so-called Coulomb
space $\cC$ is defined as  
$$
\cC:=\left\{f\in \cS'(\R^3)\ |\ D(f,f)<\ii \quad \text{ where }\quad
D(f,f):= 4\pi \int_{\R^3} \frac{|\hat f(k)|^2}{|k|^2} \, dk \right\}.
$$
Endowed with its natural inner product
$$
\langle f,g \rangle_\cC := D(f,g):= 4\pi \int_{\R^3}
\frac{\overline{\hat f(k)} \, \hat g(k)}{|k|^2} \, dk,
$$
$\cC$ is a Hilbert space. Its dual space is
$$
\cC' := \left\{ V \in L^6(\RR^3) \, | \, \nabla V \in (L^2(\RR^3))^3 \right\},
$$
endowed with the inner product
$$
\langle V_1,V_2 \rangle_{\cC'} := \frac{1}{4\pi}
\int_{\R^3} \nabla V_1 \cdot \nabla V_2 = 
\frac{1}{4\pi}\int_{\R^3} |k|^2\overline{\hat V_1(k)} \, \hat V_2(k)\, dk.
$$
Note that if $Q \in \cK \cap \gS_1$, then of course $\tr_0(Q)=\tr(Q)$,
$\rho_Q(\cdot)=Q(\cdot,\cdot) \in L^1(\R^3)$ and $\tr(Q) = \int_{\R^3} \rho_Q$. 

The energy functional $E_{\nu,\epsilon_{\rm F}}$ is well-defined on
$\cK$ for all $\nu$ such that $(\nu \star |\cdot|^{-1}) \in L^2(\RR^3) +
\cC'$. The first term of $E_{\nu,\epsilon_{\rm F}}$ makes sense as it holds 
$$
c_1 (1-\Delta) \le |H^0_{\rm per}-\epsilon_{\rm F}| \le C_1 (1-\Delta)
$$
for some constants $0 < c_1 < C_1 < \infty$
(see \cite[Lemma 1]{CanDelLew-08a}). The last two terms of $E_{\nu,\epsilon_{\rm F}}$ are
also well defined since 
$\rho_Q \in L^2(\RR^3) \cap \cC$ for all $Q \in \cK$.

The following
result is a straightforward extension of Theorem~2 in
\cite{CanDelLew-08a}, allowing in particular to account for point-like
nuclar charges: if $\nu$ is a Dirac mass, $\nu \star |\cdot|^{-1} \in
L^2(\R^3) + \cC'$. 

\begin{theorem}[Existence of a minimizer for perturbed crystal] \label{thm:defaut}
Let $\nu$ such that $(\nu \star |\cdot|^{-1}) \in L^2(\RR^3) +
\cC'$. Then, the minimization problem
\begin{equation} \label{eq:min_E}
\inf \left\{ E_{\nu,\epsilon_{\rm F}}(Q), \; Q \in \cK \right\}
\end{equation}
has a minimizer $Q_{\nu,\epsilon_{\rm F}}$, and all the minimizers of
(\ref{eq:min_E}) share the same density $\rho_{\nu,\epsilon_{\rm
    F}}$. In addition, $Q_{\nu,\epsilon_{\rm F}}$ is solution to the
self-consistent equation
\begin{equation} \label{eq:SCF}
Q_{\nu,\epsilon_{\rm F}} = 
1_{(-\infty,\epsilon_{\rm F})} \left(H^0_{\rm per} + (\rho_{\nu,\epsilon_{\rm
    F}}-\nu) \star |\cdot|^{-1} \right) -  
1_{(-\infty,\epsilon_{\rm F}]} \left(H^0_{\rm per}\right) + \delta
\end{equation}
where $\delta$ is a finite-rank self-adjoint operator on $L^2(\R^3)$
such that $0 \le \delta \le 1$ and $\mbox{Ran}(\delta) \subset
\mbox{Ker}\left( H^0_{\rm per} - \epsilon_{\rm F} \right)$. 
\end{theorem}

\begin{remark}
Our notation $Q_{\nu,\epsilon_{\rm F}}$ does not mean that minimizers of
$E_{\nu,\epsilon_{\rm F}}$ are necessarily uniquely defined (although
the minimizing density $\rho_{\nu,\epsilon_{\rm F}}$ is itself unique). However, as we will see below in Lemma \ref{prop:small}, when $\nu\ll1$ in an appropriate sense, one has $\delta=0$ hence $Q_{\nu,\epsilon_{\rm F}}$ is indeed unique.
\end{remark}

In this approach,
the electronic charge of the defect is controlled by the Fermi level
$\epsilon_{\rm F}$, and not via a direct constraint on $\tr_0(Q)$ (see \cite{CanDelLew-08a} for results in the latter case).
When 
$$
\epsilon_{\rm F} \in (\Sigma_N^+,\Sigma_{N+1}^-) \setminus 
\sigma\left(H^0_{\rm per}+(\rho_{\nu,\epsilon_{\rm F}}-\nu) \star
  |\cdot|^{-1}\right), 
$$ 
the minimizer $Q_{\nu,\epsilon_{\rm F}}$ is uniquely defined. It is both a Hilbert-Schmidt
operator ($\tr(Q_{\nu,\epsilon_{\rm F}})^2<\ii$) and the difference
of two orthogonal projectors (since $\delta=0$ then). In this case,
$\tr_0(Q_{\nu,\epsilon_{\rm F}})$ is always an integer, as
proved in \cite[Lemma 2]{HaiLewSer-05a}. 
The integer $\tr_0(Q_{\nu,\epsilon_{\rm F}})$ can be interpreted as the
electronic charge of the state 
$\gamma = \gamma^0_{\rm per}+Q_{\nu,\epsilon_{\rm F}}$ (measured with
respect to the Fermi sea $\gamma^0_{\rm per}$). 
We will see later in Lemma \ref{prop:small} that
$\tr_0(Q_{\nu,\epsilon_{\rm F}})=0$ whenever $\nu$ is small
enough. 

Note that the fact that $Q_{\nu,\epsilon_{\rm F}}$ is both a Hilbert-Schmidt
operator and the difference of two orthogonal projectors automatically implies
that the generalized trace of $Q_{\nu,\epsilon_{\rm F}}$ is well-defined since 
$$
Q_{\nu,\epsilon_{\rm F}}^2=Q_{\nu,\epsilon_{\rm F}}^{++}-Q_{\nu,\epsilon_{\rm F}}^{--},
$$
with $Q_{\nu,\epsilon_{\rm F}}^{++} \ge 0$ and $Q_{\nu,\epsilon_{\rm
    F}}^{--} \le 0$. Let us remark incidently that the condition 
$\tr(Q_{\nu,\epsilon_{\rm F}}^2)<\ii$ is required in the
Shale-Stinespring Theorem which guarantees the equivalence of the Fock space
representations \cite{HaiLewSer-05a,CanDelLew-08b} defined by $\gamma^0_{\rm per}$ and
$\gamma=\gamma^0_{\rm per}+Q_{\nu,\epsilon_{\rm F}}$ respectively. 

One of the purposes of this article is to study in more details the operator $Q_{\nu,\epsilon_{\rm F}}$ and the function $\rho_{{\nu,\epsilon_{\rm F}}}$. 

\section{Linear response to an effective potential}
\label{sec:response_eff}

In this section, we study the linear response of the electronic ground state of
a crystal to a small {\em effective} potential $V \in L^2(\R^3) +
\cC'$. This means more precisely that we expand the formula
$$
Q_V = 1_{(-\infty,\epsilon_{\rm F}]} \left(H^0_{\rm per} + V \right) -  
1_{(-\infty,\epsilon_{\rm F}]} \left(H^0_{\rm per}\right),
$$
in powers of $V$ (for $V$ small enough) and state some important properties of the first order term. The higher order terms will be studied with more details later in Lemma \ref{lem:order2}. Obviously, the first order term will play a decisive role in the study of the properties of nonlinear minimizers.

As mentioned in the introduction, the results of this section can be
used not only for the reduced Hartree-Fock model considered in the paper, but also for the linear model and for the
Kohn-Sham LDA framework. In the reduced
Hartree-Fock model, the effective potential is $V =
(\rho_{\nu,\epsilon_{\rm F}}-\nu) \star |\cdot|^{-1}$. In the linear model,
the interaction between electrons is neglected and $V$ coincides with
the external potential: $V = V_{\rm ext} = - \nu \star |\cdot|^{-1}$. In
the Kohn-Sham LDA model,
$$
V = (\rho_{\nu,\epsilon_{\rm F}}^{{\rm LDA}}-\nu) \star |\cdot|^{-1} + v_{\rm
  xc}^{\rm LDA}(\rho^0_{\rm per}+ \rho_{\nu,\epsilon_{\rm F}}^{{\rm LDA}})-v_{\rm
  xc}^{\rm LDA}(\rho^0_{\rm per})
$$
where $v_{\rm xc}$ is the LDA exchange-correlation potential and
$\rho_{\nu,\epsilon_{\rm F}}^{{\rm LDA}}$ the variation of the 
electronic density induced by the external potential $V_{\rm ext} = -
\nu \star |\cdot|^{-1}$, see \cite{CanDelLew-08b}. 

\bigskip

Expanding (formally) $Q_V$ in powers of $V$ and using the resolvent formula leads to considering the following operator
\begin{equation}
Q_{1,V} =  \frac{1}{2i\pi} \oint_\curv
\left( z-H^0_{\rm per} \right)^{-1} V \left(z-H^0_{\rm
      per}\right)^{-1} \, dz,
\label{eq:Q1V} 
\end{equation}
where  $\curv$ is a smooth curve in the
complex plane enclosing the whole 
spectrum of $H^0_{\rm per}$ below $\epsilon_{\rm F}$, crossing the
real line at $\epsilon_{\rm F}$ and at some $c < \inf\sigma(H^0_{\rm
  per})$. In order to relate our work to the Physics literature, we start by defining the independent particle polarizability operator $\chi_0$.

\begin{proposition}[Independent particle polarizability]\label{prop:chi0} 
If $V \in L^2(\R^3) + \cC'$, the operator $Q_{1,V}$ defined above in
\eqref{eq:Q1V} 
is in $\cQ$ and $\tr_0(Q_{1,V})=0$. If $V \in L^1(\R^3)$,
$Q_{1,V} \in \gS_1$ and $\tr(Q_{1,V})=0$.

The independent particle
polarizability operator $\chi_0$ defined by 
$$\boxed{\chi_0 V := \rho_{Q_{1,V}}}$$
is a continuous linear application from $L^1(\R^3)$ to $L^1(\R^3)$ and 
from $L^2(\R^3) + \cC'$ to $L^2(\R^3) \cap \cC$.
\end{proposition}

The proof of Proposition \ref{prop:chi0} is provided below in Section \ref{sec:proof_chi0}.

For the cases we have to deal with, we can consider that the effective potential $V$ is the Colomb potential generated by a charge
distribution $\rho$:
$$V=\rho\star|\cdot|^{-1} \quad \text{for some $\rho$}.$$
We will have $\rho=-\nu$ for a linear model (non-interacting electrons)
and $\rho=\rho_{{\nu,\epsilon_{\rm F}}}-\nu$ for the nonlinear reduced
Hartree-Fock model. Following usual Physics notation, we denote by 
 $v_{\rm c}$ the Coulomb operator: 
$$
v_{\rm c}(\rho) = \rho \star |\cdot|^{-1},
$$ 
which defines an isometry from $\cC$ onto $\cC'$.
If $\rho \in v_{\rm c}^{-1}(L^2(\R^3)+\cC')=v_{\rm
  c}^{-1}(L^2(\R^3))+\cC$, then we have $v_{\rm c}(\rho) \in L^2(\R^3) +
\cC'$, hence $Q_{1,v_{\rm c}(\rho)}\in\cQ$ and $\rho_{Q_{1,v_{\rm c}(\rho)}}\in L^2(\R^3)\cap\cC$. 

We now define the linear response operator
\begin{equation}
\boxed{\cL(\rho):=-\rho_{Q_{1,v_{\rm c}\rho}}}
\label{def:cC}
\end{equation}
and we concentrate on the study of the operator $\cL$. As
$\cL=-\chi_0v_{\rm c}$, it follows from Proposition~\ref{prop:chi0} that
$\cL$ maps $\cC$ into $\cC \cap L^2(\R^3)$.
The reason why we have put a minus sign is very simple: in the rHF nonlinear case, we will have 
$$\rho_{\nu,\epsilon_{\rm F}}=\cL(\nu-\rho_{\nu,\epsilon_{\rm F}})+\tilde{r}_2$$
where $\tilde{r}_2$ contains the higher order terms, and and we will
rewrite the above equality as 
\begin{equation}
 (1+\cL)(\nu-\rho_{\nu,\epsilon_{\rm F}})=\nu-\tilde{r}_2.
\label{eq:SCF_rewritten}
\end{equation}
This motivates the following result.

\begin{proposition}[Self-adjointness of the operator $\cL$] \label{lem:defL}
$\cL$ defines a bounded nonnegative self-adjoint operator on $\cC$. 
Hence $1+\cL$, considered as an operator on $\cC$, is invertible and
bicontinuous from $\cC$ to $\cC$. 
\end{proposition}

The latter property will be used in Section \ref{sec:response_ext}.

\medskip

We have considered the linear response for all reasonable $V$'s (or $\rho$'s). We now assume that $V=\rho\star
|\cdot|^{-1}$ with a density $\rho\in L^1(\R^3)$ and we derive some
additional properties of $\cL(\rho)$. Note that as $L^1(\R^3) \subset
v_{\rm c}^{-1}(L^2(\R^3)) + \cC$, we have $v_{\rm c}(\rho) \in L^2(\R^3)
+ \cC'$. The following statement is central in the mathematical analysis
of the dielectric response of crystals.  

\begin{proposition}[Properties of $\cL(\rho)$ when $\rho\in L^1$] \label{lem:limL}
Let $\rho \in L^1(\RR^3)$. Then, $\cL(\rho) \in L^2(\R^3) \cap \cC$,
$\widehat{\cL(\rho)}$ is continuous on $\R^3 \setminus
\cR^\ast$, and for all $\sigma \in S^2$ (the unit sphere of $\R^3$), 
\begin{equation} \label{eq:limL}
\lim_{\eta \to 0^+} \widehat{\cL(\rho)}(\eta \sigma) = (\sigma^T L
\sigma) \widehat\rho(0) 
\end{equation}
where $L \in \R^{3 \times 3}$ is the non-negative symmetric matrix
defined by
\begin{equation} \label{eq:mk}
\forall k \in \R^3, \quad 
k^TLk  =   \frac{8\pi}{|\Gamma|}
\sum_{n=1}^N\sum_{n'=N+1}^{+\ii}\fint_{\Gamma^\ast}\frac{\left|\pscal{(k\cdot\nabla_x)u_{n,q},u_{n',q}}_{L^2_{\rm
        per}(\Gamma)}\right|^2}{\big(\epsilon_{n',q}-\epsilon_{n,q}\big)^3} \, dq, 
\end{equation}
where the $\epsilon_{n,q}$'s and the $u_{n,q}$'s are the eigenvalues and
the eigenvectors arising in the spectral decomposition \eqref{eq:dec_HOperq} of
$(H^0_{\rm per})_q$. 

Additionally, 
\begin{equation} \label{eq:Lpositive}
L_0 = \frac 13 \tr(L) > 0.
\end{equation}
\end{proposition}

\medskip

Proposition \ref{lem:limL} shows that $\cL(\rho)$ is not in 
general a function of $L^1(\R^3)$ even when $\rho\in L^1(\R^3)$, as when
$L \neq L_0$ (i.e. when $L$ is not proportional to the identity matrix),
$\widehat{\cL(\rho)}$ is not continuous at zero (note 
that $L= L_0$ characterizes isotropic dielectric
materials). However the following holds: for any radial 
function $\xi\in C^\ii_0(\R^3)$ such that $0\leq \xi\leq 1$, $\xi\equiv
1$ on $B(0,1)$ and $\xi\equiv 0$ on $\R^3\setminus B(0,2)$, we have
\begin{equation}
\lim_{R\to\ii}\int_{\R^3} \cL(\rho)(x) \, \xi(R^{-1}x)\, dx=L_0
\int_{\R^3}\rho. 
\label{integral_rho_radial} 
\end{equation}

\section{Application to the reduced Hartree-Fock model for perturbed crystals}
\label{sec:response_ext}

Let us now come back to the reduced Hartree-Fock framework and the decay properties of minimizers. Our main result is the following
\begin{theorem}[Properties of the nonlinear rHF ground state for perturbed crystals] \label{thm:notL1}
Let $\nu \in L^1(\R^3)\cap L^2(\R^3)$ be such that $\int_{\R^3} \nu \neq 0$ and  $\| \nu \star |\cdot|^{-1}\|_{L^2+\cC'}$ is small enough.
Then the operator $Q_{\nu,\epsilon_{\rm F}}$ satisfies $\tr_0(Q_{\nu,\epsilon_{\rm F}})=0$ but it is \underline{not} trace-class.
If additionally the 
map $L:S^2\to\R^+$ defined in Proposition \ref{lem:limL} is not constant, then $\rho_{{\nu,\epsilon_{\rm F}}}$ is \underline{not} in $L^1(\R^3)$.
\end{theorem}

The proof of Theorem \ref{thm:notL1} is a simple consequence of our
results on the operator $\cL$ stated in the last section, and of the
continuity properties of higher order terms for an $L^1$ density
$\rho$. A detailed proof is provided in Section \ref{sec:proof_notL1}.

As previously mentioned, the situation $L = L_0$ characterizes
isotropic dielectric materials; it occurs in 
particular when $\cR$ is a cubic lattice and $\rho^{\rm nuc}_{\rm per}$
has the symmetry of the cube. For anisotropic dielectric materials, $L$
is not proportional to the identity matrix, and consequently
$\rho_{{\nu,\epsilon_{\rm F}}} \notin L^1(\R^3)$. 

Formula \eqref{eq:Lpositive} for $L_0$ is well-known in the Physics
literature \cite{Adler-62,Wiser-63}. However to our knowledge it was never
mentioned that the fact that $L_0>0$ is linked to the odd mathematical
property that the operator $Q_{\nu,\epsilon_{\rm F}}$ is not trace-class
when $\int_{\R^3} \nu \neq 0$.
The interpretation is the following: if a defect
with nuclear charge $\nu$ is inserted in the crystal, the Fermi sea
reacts to the modification of the
external potential. Although it stays formally neutral
($\tr_0(Q_{\nu,\epsilon_{\rm F}})=0$) when $\nu$ is small, the
modification $\rho_{\nu,\epsilon_{\rm F}}$ of the electronic density
generated by $\nu$ is not an integrable function such that $\int_{\R^3}
\rho_{\nu,\epsilon_{\rm F}} = 0$, as soon as $\int_{\R^3} \nu
\neq 0$. 

For isotropic dielectric materials, $L=L_0$, and we conjecture that the
density $\rho_{\nu,\epsilon_{\rm F}}$ 
is in $L^1(\R^3)$. In this case, one can define the total charge of the
defect (including the self-consistent polarization of the Fermi sea) as
$\int_{\R^3}(\nu-\rho_{\nu,\epsilon_{\rm F}})$. For $\nu$ small enough, the
Fermi sea formally stays neutral
($\tr_0(Q_{\nu,\epsilon_{\rm F}})=0$), but it nevertheless screens
partially the charge defect in such a way that the total observed charge
gets multiplied by a factor
$(1+L_0)^{-1}<1$:
$$
\int_{\R^3} \big(\nu-\rho_{\nu,\epsilon_{\rm F}}\big) = \frac{\int_{\R^3} \nu}{1+L_0}.
$$
This is very much
similar to what takes place in the mean-field approximation of no-photon QED
\cite{HaiLewSerSol-07,GraLewSer-08}. In the latter setting, the Dirac
sea screens any external charge, leading to charge
renormalization. Contrarily to the Fermi sea of periodic crystals, the
QED free vacuum is not only isotropic (the corresponding
$L$ is proportional to the identity matrix) but also homogeneous (the
corresponding 
operator $\cL$ has a simple expression in the Fourier representation),
and the mathematical 
analysis can be pushed further: Gravejat, Lewin and Séré indeed proved in
\cite{GraLewSer-08} that the observed electronic density in QED (the
corresponding $\rho_{\nu,\epsilon_{\rm F}}$) actually belongs to $L^1(\R^3)$. 
Extending these results to the case of isotropic crystals seems to be a challenging task.

When $L$ is not proportional to the identity matrix (anisotropic dielectric
crystals), it is not possible to define the observed charge of
the defect as the integral of $\nu-\rho_{{\nu,\epsilon_{\rm F}}}$
since $\rho_{{\nu,\epsilon_{\rm F}}}$ is not an integrable
function. Understanding the regularity properties of the Fourier
transform of $\rho_{\nu,\epsilon_{\rm F}}$ is then a very interesting
problem. In the next section, we consider a certain limit related to
homogenization in which only the first order term plays a role and for
which the limit can be analyzed in details.

\section{Macroscopic dielectric permittivity}
\label{sec:macroscopic}

In this section, we focus on the electrostatic potential 
\begin{equation} \label{eq:tot_pot}
V = (\nu-\rho_{\nu,\epsilon_{\rm F}}) \star |\cdot|^{-1}
\end{equation}
generated by the total charge of the defect and we study it in a certain limit.

We note that the self-consistent equation \eqref{eq:SCF_rewritten} can be rewritten as 
\begin{equation}
\nu-\rho_{\nu,\epsilon_{\rm F}}=(1+\cL)^{-1}\nu-(1+\cL)^{-1}\tilde{r}_2.
\label{eq:SCF_rewritten2}
\end{equation}
Therefore for the nonlinear rHF model, the linear response at the level of the density is given by the operator $(1+\cL)^{-1}$. We recall from Proposition \ref{lem:defL} that $\cL\geq0$ on $\cC$ and that $(1+\cL)^{-1}$ is a bounded operator from $\cC$ to $\cC$.  

In Physics, one is often interested in the \emph{dielectric
  permittivity} which is the inverse of the linear response at the level
of the electrostatic potential, i.e. 
$$\boxed{\epsilon^{-1}:=v_{\rm c}(1+\cL)^{-1}v_c^{-1}.}$$
Note that \eqref{eq:SCF_rewritten} can be recast into
\begin{equation} \label{eq:tot_pot_2}
V  =  \epsilon^{-1}v_{\rm c}(\nu) - v_c(1+\cL)^{-1}\tilde{r}_2.
\end{equation}
A simple calculation gives (we recall that $\chi_0$ is the polarizability defined in Proposition \ref{prop:chi0}, which is such that $\cL=-\chi_0 v_{\rm c}$)
$$
\epsilon^{-1}v_{\rm c} =v_{\rm c}(1+\cL)^{-1}= v_{\rm c} (1+\cL)^{-1} (1+\cL+\chi_0v_{\rm c})=  v_{\rm c} + v_{\rm c} (1+\cL)^{-1} \chi_0v_{\rm c}.
$$
Therefore one gets
\begin{equation}
\epsilon^{-1}=1 + v_{\rm c} (1+\cL)^{-1} \chi_0.
\label{def_epsilon_inverse} 
\end{equation}
We also have 
\begin{equation}
\epsilon=v_{\rm c}(1+\cL)v_c^{-1}
\label{def:dielectric} 
\end{equation}
which yields to the usual formula
\begin{equation}
\epsilon=1 -v_c\chi_0.
\label{def_epsilon} 
\end{equation}

The basic mathematical properties of the dielectric operator $\epsilon$ are stated in the following
\begin{proposition}[Dielectric operator]\label{thm:dielectric} 
The dielectric operator $\epsilon=1-v_{\rm c}\chi_0$ is an 
invertible bounded self-adjoint operator on $\cC'$, with inverse $\epsilon^{-1} = 1 + v_{\rm c} (1+\cL)^{-1} \chi_0$.

The hermitian dielectric operator $\widetilde \epsilon = v_{\rm
  c}^{-\frac 12} \epsilon v_{\rm c}^{\frac 12}$ is an invertible bounded
self-adjoint operator on $L^2(\R^3)$.
\end{proposition}

The proof of Proposition is a simple consequence of the properties of $\chi_0$ and $\cL$, as explained in Section \ref{sec:proof_dielectric}.

\medskip

Even when $\nu\in L^1(\R^3)$, applying the operator $(1+\cL)^{-1}$
creates some discontinuities in the Fourier domain for the corresponding
first order term $(1+\cL)^{-1}\nu$ in Equation
\eqref{eq:SCF_rewritten2}. If we knew that the higher order term
$\tilde{r}_2$ was better behaved, it would be possible to deduce the
exact regularity of $\widehat{\rho}_{\nu,\epsilon_{\rm F}}$. We will now
consider a certain limit of (\ref{eq:SCF_rewritten2}) by means of a
homogenization argument, for which the higher order term
vanishes. This will give an illustration of the expected properties of
the density in Fourier space at the origin. 
For this purpose, we fix some $\nu \in L^1(\R^3) \cap L^2(\R^3)$ and
introduce for all $\eta > 0$ the rescaled density 
$$
\nu_\eta(x) := \eta^3 \nu(\eta x).
$$
We then denote by $V_\nu^\eta$ the total potential generated by $\nu_\eta$,
i.e. 
\begin{equation} \label{eq:Vnueta}
V_\nu^\eta := (\nu_\eta-\rho_{\nu_\eta,\epsilon_{\rm F}})\star |\cdot|^{-1},
\end{equation}
and define the rescaled potential
\begin{equation} \label{eq:Wnueta}
W_\nu^\eta(x) := \eta^{-1}  \, V_\nu^\eta \left( \eta^{-1} x  \right).
\end{equation}
Note that the scaling parameters have been chosen in such a way that in
the absence of dielectric response (i.e. for $\epsilon^{-1}=1$,
$\tilde r_2=0$), one has $W_\nu^\eta = v_{\rm c}(\nu) = \nu \star
|\cdot|^{-1}$ for all $\eta > 0$. 

\medskip

\begin{theorem}[Macroscopic Dielectric Permittivity]\label{thm:homogenization}
There exists a $3 \times 3$ symmetric matrix $\epsilon_{\rm M} \ge 1$ such
that for all $\nu \in L^1(\R^3) \cap L^2(\R^3)$, the rescaled potential 
$W_\nu^\eta$ defined by (\ref{eq:Wnueta}) converges to $W_\nu$ weakly in $\cC'$ when $\eta$ goes to zero, where $W_\nu$ is the unique
solution in $\cC'$ to the elliptic equation
$$
\boxed{-\div(\epsilon_{\rm M} \nabla W_\nu ) = 4\pi\nu.}
$$
The matrix $\epsilon_{\rm M}$ is proportional to the identity matrix if the
host crystal has the symmetry of the cube.
\end{theorem}

\medskip

From a physical viewpoint, the matrix $\epsilon_{\rm M}$ is the
electronic contribution to the macroscopic dielectric tensor of 
the host crystal. Note the other contribution, originating from the displacements of the
nuclei~\cite{PicCohMar-70}, is not taken into account in our study. 

The matrix $\epsilon_{\rm M}$ can be computed from the Bloch-Floquet
decomposition of $H^0_{\rm per}$ as follows.
The operator $\widetilde \epsilon^{-1} = v_{\rm c}^{-1/2}
\epsilon^{-1} v_{\rm c}^{1/2}$ commuting with the
translations of the lattice, i.e. with $\tau_R$ for all $R \in \cR$,
it can be represented by the Bloch matrices
$([\widetilde \epsilon_{KK'}^{-1}(q)]_{K,K' \in \cR^\ast})_{q \in
  \Gamma^\ast}$:
$$
\forall f \in L^2(\R^3), \quad 
\widehat{\widetilde \epsilon^{-1} f}(q+K) = 
\sum_{K' \in \cR^\ast} \widetilde \epsilon_{KK'}^{-1}(q) \widehat f(q+K')
$$
for almost all $q \in \Gamma^\ast$ and $K \in \cR^\ast$. 
We will show later in Lemma \ref{lem:def_epsilon} that
$\tilde{\epsilon}_{K,K'}(\eta\sigma)$ has a limit when $\eta$ goes to
$0^+$ for all fixed $\sigma\in S^2$. Indeed one has  
$$
\lim_{\eta\to0^+}\tilde{\epsilon}_{0,0}(\eta\sigma)=1+\sigma^T L \sigma
$$
where $L$ is the $3 \times 3$ non-negative symmetric matrix defined
in~\eqref{eq:mk}.  
When $K,K'\neq0$, $\tilde{\epsilon}_{K,K'}(\eta\sigma)$ has a limit at
$\eta=0$, which is independent of $\sigma$ and which we simply denote as
$\tilde{\epsilon}_{K,K'}(0)$. When $K=0$ but $K'\neq0$, the limit is a
linear function of $\sigma$: for all $K' \in \cR^\ast \setminus
\left\{0\right\}$,  
$$
\lim_{\eta\to0^+} \tilde{\epsilon}_{0,K'}(\eta\sigma)=\beta_{K'} \cdot \sigma,
$$
for some $\beta_{K'} \in \CC^3$.
The electronic contribution to the macroscopic
dielectric permittivity is the $3 \times 3$ symmetric tensor
defined as \cite{BarRes-86}
\begin{equation} \label{eq:epsilonM}
\forall k \in \R^3, \quad 
k^T \epsilon_{\rm M} k = \lim_{\eta \to 0^+} \frac{|k|^2}{[
  \widetilde\epsilon^{-1}]_{00}(\eta k)}.
\end{equation}  
By the Schur complement formula, one has
$$ \frac{1}{[\widetilde\epsilon^{-1}]_{00}(\eta
  k)}=\tilde\epsilon_{00}(\eta
k)-\sum_{K,K'\neq0}{\tilde{\epsilon}_{0,K}(\eta k)}[C(\eta
k)^{-1}]_{K,K'}\tilde{\epsilon}_{K',0}(\eta k) 
$$
where $C(\eta k)^{-1}$ is the inverse of the matrix
$C(\eta k)=[\tilde\epsilon_{KK'}(\eta k)]_{K,K'\in \cR^\ast\setminus\left\{0\right\}}$. 
This leads to
$$ \lim_{\eta\to0^+}\frac{|k|^2}{[\widetilde\epsilon^{-1}]_{00}(\eta k)}
=|k|^2+k^TLk-\sum_{K,K' \in \cR^\ast\setminus\left\{0\right\}} (\beta_{K}
\cdot k) [C(0)^{-1}]_{K,K'}  
\overline{(\beta_{K'} \cdot k)}$$
where $C(0)^{-1}$ is the inverse of the matrix
$C(0)=[\tilde\epsilon_{KK'}(0)]_{K,K'\in
  \cR^\ast\setminus\left\{0\right\}}$, hence to
\begin{equation} \label{eq:epsilon_M}
\boxed{\epsilon_{\rm M} = 1 + L -\sum_{K,K'\in
  \cR^\ast\setminus\left\{0\right\}} \beta_{K} [C(0)^{-1}]_{K,K'}
\beta_{K'}^\ast.}
\end{equation}
As already noticed in \cite{BarRes-86}, it holds
$$
1 \leq \epsilon_{\rm M} \leq 1 + L.
$$

Formula~\eqref{eq:epsilon_M} has been used in numerical simulations for
estimating 
the macroscopic dielectric permittivity of real insulators and
semiconductors~\cite{BarRes-86,HybLou-87a,HybLou-87b,EngFar-92,GajHumKreFurBec-06}.
Direct methods for evaluating $\epsilon_{\rm M}$, bypassing the
inversion of the matrix $C(0)$, have also been
proposed~\cite{ResBal-81,KunTos-84}.

\section{Proofs} 
\label{sec:proofs}

In this last section, we gather the proofs of all the results of this paper.

\subsection{Preliminaries}
Let us first recall some useful results established in
\cite{CanDelLew-08a}.
\begin{lemma}[Some technical estimates] \label{lem:technical} Let $\Lambda$ be a compact subset of $\C\setminus\sigma(H^0_{\rm per})$.
\begin{enumerate}
\item The operator $B(z):=(z-H^0_{\rm per})^{-1}(1-\Delta)$ and its
  inverse are bounded uniformly on $\Lambda$.
\item The operators $|\nabla|\times|z-H^0_{\rm per}|^{-\frac 12}$ and
  $|\nabla|(z-H^0_{\rm per})^{-1}$ are bounded uniformly on $\Lambda$.
\item There exists two positive constants $0 < c_1 < C_1 < \infty$ such
  that
\begin{equation} \label{eq:bounds_H0per}
c_1 (1-\Delta) \le |H^0_{\rm per}-\epsilon_{\rm F}| \le C_1 (1-\Delta).
\end{equation}
\item If $V \in L^2(\R^3) + \cC'$, $[\gamma^0_{\rm per},V] \in \gS_2$
  and there exists a constant $C \in \R_+$ independent of $V$ such that
\begin{equation} \label{eq:commutator1}
\|[\gamma^0_{\rm per},V]\|_{\gS_2} \le C \|V\|_{L^2(\R^3) + \cC'}.
\end{equation}
Besides, if $V \in L^q(\R^3)$ for some $1 \le q \le \infty$ and if
$\nabla V \in L^p(\R^3)$ for some $\frac 65 <p< \ii$, then
$$
\left\|[\gamma^0_{\rm per},V]\right\|_{\gS_p}\leq C\|\nabla V\|_{L^p(\R^3)}.
$$
\end{enumerate}
\end{lemma}

We denote as usual by $\gS_p$ the space of all operators $A$ such
that $\tr(|A|^p)<\ii$, endowed with the norm
$\|A\|_{\gS_p}:=\tr(|A|^p)^{\frac 1p}$. 

\begin{proof}
For (1), (2), (3) and the first assertion of (4), see the proofs of
\cite[Lemma~1]{CanDelLew-08a} and 
  \cite[Lemma~3]{CanDelLew-08a}. The last estimate is obtained like in \cite[p. 148]{CanDelLew-08a} by writing\footnote{Note there is a sign misprint in the corresponding formula at the top of p. 148 in \cite{CanDelLew-08a}.}
\begin{multline*}
[\gamma^0_{\rm per},V]=-\sum_{j=1}^3 \frac{1}{4 i \pi} \int_\curv B(z)
\left((-\Delta+1)^{-1} \partial_{x_j} \right)
\frac{\partial V}{\partial{x_j}} (-\Delta+1)^{-1}B(z)^\ast\, dz \\
- \sum_{j=1}^3 \frac{1}{4 i \pi} \int_\curv
B(z)(-\Delta+1)^{-1} \frac{\partial V}{\partial{x_j}} \left( 
\partial_{x_j} (-\Delta+1)^{-1} \right)) B(z)^\ast\, dz.
\end{multline*}
It then suffices to use the Kato-Seiler-Simon inequality (see \cite{SeiSim-75} and
\cite[Thm 4.1]{Simon-79}) 
\begin{equation}
\forall p\geq2,\qquad \|f(-i\nabla)g(x)\|_{\gS_p}\leq (2\pi)^{-\frac 3p}
\|g\|_{L^p(\R^3)}\|f\|_{L^p(\R^3)}
\label{KSS}
\end{equation}
and the fact that $B(z)$ is uniformly bounded on $\curv$.
\end{proof}

\subsection{Proof of Theorem~\ref{thm:defaut}}
Let $\nu$ be such that $V = (\nu
\star |\cdot|^{-1}) \in L^2(\R^3) + \cC'$. As $C^\infty_0(\R^3)$ is
dense in $L^2(\R^3)$ and is included in $\cC'$, $V$ can be decomposed
for all $\eta > 0$ as $V=V_{2,\eta} + V'_\eta$ with
$V_{2,\eta} \in L^2(\R^3)$, $V'_\eta \in \cC'$ and
$\|V_{2,\eta}\|_{L^2} \le \eta$. Denoting by $\nu'_\eta = -
(4\pi)^{-1} \Delta V'_\eta$, we obtain $\nu'_\eta \in \cC$ and
$$
\forall Q \in \cQ, \quad - \int_{\RR^3} \rho_Q V  \ge 
- \eta \|\rho_Q\|_{L^2} - D(\nu'_\eta,\rho_Q).
$$
By \cite[Prop.~1]{CanDelLew-08a}, we know that there
exists a constant $C \in \RR_+$ such that
\begin{equation*}
\forall Q \in \cQ, \quad \|\rho_Q\|_{L^2} \le  C \|Q\|_\cQ.
\end{equation*}
Besides, for all $Q \in \cK$, $Q^2 \le Q^{++}-Q^{--}$ with $Q^{++} \ge
0$ and $Q^{--} \le 0$. Hence, 
$$
\forall Q \in \cK, \quad \|\rho_Q\|_{L^2} \le C'+C'
 \tr\left((1-\Delta)(Q^{++}-Q^{--})\right).
$$
Using (\ref{eq:bounds_H0per}) and choosing $\eta > 0$ such that $2\eta C'
< c_1$ leads to   
\begin{equation} \label{eq:bounds}   
\forall Q \in \cK, \quad E_{\nu,\epsilon_{\rm F}}(Q) \ge \frac{c_1}2
\tr\left((1-\Delta)(Q^{++}-Q^{--})\right) - C' \eta - \frac
12D(\nu'_\eta,\nu'_\eta). 
\end{equation}
The above inequality provides the bounds on the minimization sequences
of (\ref{eq:min_E}) which allow one to complete the proof of
Theorem~\ref{thm:defaut} by transposing the arguments used in the proof
of  \cite[Theorem~2]{CanDelLew-08a}. \qed

\subsection{Expanding $Q_V$}\label{sec:expansion_Q_V}
In this section, we explain in details how to expand
$$Q_V:=1_{(-\infty,\epsilon_{\rm F}]} \left(H^0_{\rm per} + V \right) -  
1_{(-\infty,\epsilon_{\rm F}]} \left(H^0_{\rm per}\right),$$
and give the properties of each term in the expansion.
The multiplicative operator associated with some $V \in L^2(\R^3) + \cC'$ is
a compact perturbation of $H^0_{\rm per}$, so that the operator
$H^0_{\rm per}+V$ is self-adjoint on $L^2(\R^3)$.
When $V$ is small,
it is then possible to expand $Q_V$ in a perturbative series, using the resolvent formula. For this purpose, we consider a smooth curve $\curv$ in the
complex plane enclosing the whole 
spectrum of $H^0_{\rm per}$ below $\epsilon_{\rm F}$, crossing the
real line at $\epsilon_{\rm F}$ and at some $c < \inf\sigma(H^0_{\rm
  per})$. We furthermore assume that 
$$
d(\sigma(H^0_{\rm per}),\Lambda) = \frac g 4 \quad \mbox{where} \quad 
\Lambda = \left\{z\in \CC \; | \; d(z,\curv) \le \frac g4 \right\},
$$
$d$ denoting the Euclidian distance in the complex plane and $g$ the
band gap (see Fig.~\ref{fig:contour}). 

\begin{figure}[h] 
\centering
\includegraphics{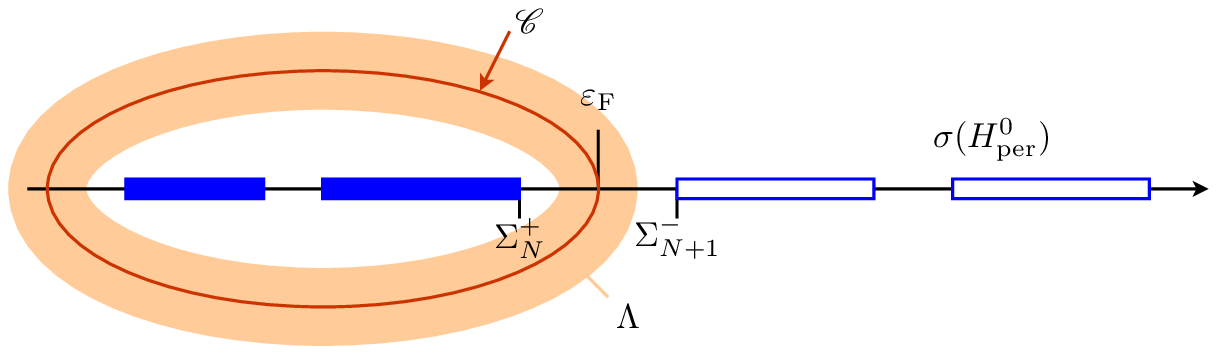}
\caption{Graphical representation of a contour $\curv \subset \CC$
  enclosing $\sigma(H^0_{\rm per}) \cap (-\infty,\epsilon_{\rm F}]$ and
  of the compact set $\Lambda$.}
\label{fig:contour}
\end{figure}

The following result will be useful to expand $Q_V$: 
\begin{lemma} \label{prop:Bogoliubov} There exists $\alpha > 0$ such that if
  $V \in L^2(\R^3) + \cC'$ is such that
$$
\| V \|_{L^2+\cC'} < \alpha,
$$
then 
\begin{equation} \label{eq:ess_spec}
\sigma \left(H^0_{\rm per} + V \right) \cap \Lambda = \emptyset.
\end{equation}
Moreover $\epsilon_{\rm F} \notin \sigma\left(H^0_{\rm per} + V
\right)$, $Q_V \in \cQ$ and $\tr_0(Q_V)=0$. 

Besides, there exists an orthonormal basis $(\phi_i^-)_{i \ge
  1}$ of the occupied space ${\mathcal H}_- = \mbox{Ran}(\gamma^0_{\rm
  per})$ and an orthonormal basis $(\phi_i^+)_{i \ge
  1}$ of the occupied space ${\mathcal H}_+ = \mbox{Ker}(\gamma^0_{\rm
  per})$ such that in the orthonormal basis $((\phi_i^-),(\phi_i^+))$ of
$L^2(\R^3) =  {\mathcal H}_- \widehat{\oplus} {\mathcal H}_+ $,
\begin{equation} \label{eq:decQ}
Q_V  =  \left( 
\begin{array}{c||c}
\mbox{diag}(-a_1,-a_2, \cdots ) &  \mbox{diag}(b_1,b_2, \cdots ) \\
\hline \hline 
\mbox{diag}(b_1,b_2, \cdots ) &  \mbox{diag}(a_1,a_2, \cdots ) \\
\end{array}
\right) 
\end{equation}
with
$$
0 \le a_i < 1, \quad \sum_{i=1}^{+\infty} a_i < \infty, \quad 
b_i = \sqrt{a_i(1-a_i)}.
$$
\end{lemma}

\medskip

The meaning of (\ref{eq:ess_spec}) is the following. As mentioned above,
any $V \in L^2(\R^3) + \cC'$ defines a compact perturbation of $H^0_{\rm
  per}$; hence the essential spectrum of the Hamiltonian
remains unchanged: 
$$
\sigma_{\rm ess}\left( H^0_{\rm
    per}+V\right)=\sigma_{\rm ess}\left( H^0_{\rm 
    per}\right).
$$
This in particular means that only eigenvalues of finite multiplicity
may appear in the gap 
$(\Sigma_N^+,\Sigma_{N+1}^-)$, and they can only accumulate at
$\Sigma_N^+$ or $\Sigma_{N+1}^-$. For $V$ small enough in $L^2+\cC'$, these
eigenvalues will be 
localized at the edges of the gap, i.e. in a vicinity of $\Sigma_N^+$
and $\Sigma_{N+1}^-$. It can be seen that the charge $\tr_0(Q_{tV})$ jumps as $t$ is increased when an eigenvalue crosses the curve $\curv$ and that it is a constant integer when this does not happen. By continuity, we deduce that $\tr_0(Q_{V})=0$: for $V$ small enough, no electron-hole pair is created from the Fermi sea.

The representation (\ref{eq:decQ}) of $Q_V$ was proved in \cite{HaiLewSer-08} and it can be
interpreted in terms of Bogoliubov states. Each $2 \times 2$
submatrix $\left( \begin{array}{cc} -a_i & b_i \\ b_i & a_i \end{array}
\right)$ can be seen as a {\em virtual} electron-hole pair. A {\em real}
electron-hole pair would be observed for $a_i=1$. 
It is easy to see that
the eigenvalues of $Q_{V}$ (including multiplicities)
are $(-a_i^{\frac 12},a_i^{\frac 12})_{i \ge 1}$. Thus a necessary and
sufficient condition for $Q_{V}$ being trace class
reads $\sum a_i^{\frac 12} < \infty$. 

\medskip

We now provide the 
\begin{proof}[Proof of Lemma \ref{prop:Bogoliubov}]
Let $C \in \RR_+$ be such that for all $z \in \Lambda$, $\|B(z)\| \le C$
(see the first statement of Lemma~\ref{lem:technical}). For all $z \in
\Lambda$, 
$$
z-H^0_{\rm per}-V = (z-H^0_{\rm per}) (1 - B(z)(1-\Delta)^{-1}V).
$$
As $L^2(\R^3) + \cC' \subset L^2(\R^3) + L^6(\R^3)$, it follows from the 
Kato-Seiler-Simon inequality \eqref{KSS} that there exists a constant $C'$ such that
$$
\forall V \in L^2(\R^3) + \cC', \quad \|(1-\Delta)^{-1}V\| \le C'
\|V\|_{L^2+\cC'}. 
$$
If $\|V\|_{L^2+\cC'} < (CC')^{-1}$, $\sigma(H^0_{\rm per}+V) \cap
\Lambda = \emptyset$. As $V$ defines a compact perturbation of $H^0_{\rm
  per}$, it follows from a standard continuity arguments that 
for  $\|V\|_{L^2+\cC'} < (CC')^{-1}$, the set $\sigma(H^0_{\rm per}+V) \cap
(-\infty,\epsilon_{\rm F}]$ lays inside the contour $\curv$, yielding
$$
1_{(-\infty,\epsilon_{\rm F}]}(H^0_{\rm per}+V) = 
\frac{1}{2i\pi} \oint_\curv  \left(z-H^0_{\rm per}-V\right)^{-1}  \, dz.
$$
Thus,
$$
Q_V = \frac{1}{2i\pi} \oint_\curv  \left[
\left(z-H^0_{\rm per}-V\right)^{-1} - \left(z-H^0_{\rm
    per}\right)^{-1}\right] \, dz. 
$$
Besides, still in the case when $\|V\|_{L^2+\cC'} < (CC')^{-1}$,
$$
\left(z-H^0_{\rm per}-V\right)^{-1} - \left(z-H^0_{\rm
    per}\right)^{-1} =  B(z)(1-\Delta)^{-1}V\left(z-H^0_{\rm per}-V\right)^{-1},
$$
so that
$$
\| Q_V \| \le \frac{|\curv|}{2\pi} \frac{C^2C' \|V\|_{L^2+\cC'}}{1-CC' \|V\|_{L^2+\cC'}}.
$$
We now set $\alpha=((1+(2\pi)^{-1}C|\curv|)CC')^{-1}$. For all $V \in
L^2(\R^3)+\cC'$ such that $\|V\|_{L^2+\cC'} < \alpha$, it holds
$\sigma(H^0_{\rm per}+V) \cap \Lambda = \emptyset$ and
$\| Q_V \| < 1$.
We conclude using \cite[Lemma 2]{HaiLewSer-05a} and \cite[Theorem 5]{HaiLewSer-08}.
\end{proof}

\bigskip

We are now going to expand $Q_V$ using the resolvent formula. We already know that
$$
\gamma^0_{\rm per} = 1_{(-\infty,\epsilon_{\rm F}]} \left(H^0_{\rm
    per}\right) = \frac{1}{2i\pi} \oint_\curv  \left(z-H^0_{\rm
  per}\right)^{-1}  \, dz.
$$
It now follows from the proof of Lemma~\ref{prop:Bogoliubov} that 
$$
1_{(-\infty,\epsilon_{\rm F}]} \left(H^0_{\rm
    per}+V\right) = \frac{1}{2i\pi} \oint_\curv  \left(z-H^0_{\rm
  per}-V\right)^{-1}  \, dz
$$
for all $V \in L^2(\R^3) + \cC'$ such that $\|V\|_{L^2+\cC'} <
\alpha$, yielding for such $V$'s
$$
Q_V =  \frac{1}{2i\pi} \oint_\curv \left[
\left( z-H^0_{\rm per} - V \right)^{-1} - \left(z-H^0_{\rm
  per}\right)^{-1} \right] \, dz.
$$ 

One important result of the present section is the following
\begin{lemma}[Resolvent expansion] \label{prop:Dyson}
Let $V \in  L^2(\R^3) + \cC'$ such that $\|V\|_{L^2+\cC'} <
\alpha$. Then, for all $K \in \N \setminus \left\{0\right\}$,
\begin{equation} \label{eq:Dyson_Q}
Q_V  =   Q_{1,V} + \cdots + Q_{K,V}  + \widetilde Q_{K+1,V} 
\end{equation}
where 
\begin{eqnarray}
Q_{k,V} & = & \frac{1}{2i\pi} \oint_\curv
\left( z-H^0_{\rm per} \right)^{-1} \left[ V \left(z-H^0_{\rm
      per}\right)^{-1}\right]^k \, dz \label{eq:QkV}  \\
\widetilde Q_{k,V} & = & \frac{1}{2i\pi} \oint_\curv
\left( z-H^0_{\rm per} -V \right)^{-1} \left[ V \left(z-H^0_{\rm
      per}\right)^{-1}\right]^{k} \, dz \label{eq:tQkV}
\end{eqnarray}
For all $k \ge 1$, the operator $Q_{k,V}$ is
in $\cQ$ and $\tr_0(Q_{k,V})=0$. For all $k \ge 1$, the operator
$\widetilde Q_{k,V}$ is in $\cQ$ and $\tr_0(\widetilde Q_{k,V})=0$. For
all $k \ge 6$, the operators $Q_{k,V}$ and $\widetilde Q_{k,V}$ are
trace-class and 
$\tr(\widetilde Q_{k,V}) = 0$. 
\end{lemma}

Note that by linearity, the operators $Q_{k,V}$ are well-defined for all
$V \in L^2(\R^3) + \cC'$, and not only for small $V$'s. It can in fact
be shown using the same arguments as in the proof of
Lemma~\ref{prop:Dyson} that for all $k \ge 1$,
$$
(V_1,\cdots,V_k) \mapsto  \frac{1}{2i\pi} \oint_\curv
\left( z-H^0_{\rm per} \right)^{-1} V_1 \left(z-H^0_{\rm
      per}\right)^{-1} \cdots  V_k \left(z-H^0_{\rm
      per}\right)^{-1} \, dz
$$
is a continuous $k$-linear application from $(L^2(\R^3) + \cC')^k$ to
$\cQ$. 

\medskip

Let us now detail the 
\begin{proof}[Proof of Lemma~\ref{prop:Dyson}]
It follows from the proof of Lemma~\ref{prop:Bogoliubov} that
each term of the expansion (\ref{eq:Dyson_Q}) makes sense in the space
of bounded 
operators on $L^2(\R^3)$. We now have to prove that $Q_{k,V}$ and
$\widetilde Q_{k,V}$ are in $\cQ$ and that their generalized trace is
equal to zero.
We start by noticing that $Q_V$ is indeed a minimizer for the functional
$$E(Q):=\tr(|H^0_{\rm per}-\epsilon_{\rm F}|(Q^{++}-Q^{--})-\int_{\R^3}V\rho_Q$$
on $\cK$. Theorem \ref{thm:defaut} with the nonlinear term erased then implies that $Q_V\in\cQ$.

Let us consider $Q_{1,V}$. Decomposing $V$ as $V =
V_2 + V'$ with $V_2 \in L^2(\R^3)$ and $V'
\in \cC' \subset L^6(\R^3)$, and using the Kato-Seiler-Simon inequality \eqref{KSS} and
the first assertion of Lemma~\ref{lem:technical}, 
$Q_{1,V_2} \in \gS_2$ and $Q_{1,V'} \in \gS_6$. Hence, $Q_{1,V}$ is
well-defined in $\gS_6$. A straightforward application of the
residuum formula then shows that $Q_{1,V}^{++}=Q_{1,V}^{--}=0$. As
\begin{eqnarray*}
Q_{1,V'}^{-+} & = & 
\frac1{2i\pi}\oint_\curv \gamma^0_{\rm per} (z-H^0_{\rm
  per})^{-1}V'(z-H^0_{\rm per})^{-1} (\gamma^0_{\rm per})^\perp dz
 \\
& = & \frac1{2i\pi}\oint_\curv  \gamma^0_{\rm per}  (z-H^0_{\rm per})^{-1}
[\gamma^0_{\rm per},V'] (z-H^0_{\rm per})^{-1} (\gamma^0_{\rm per})^\perp dz, 
\end{eqnarray*}
we can make use of Lemma~\ref{lem:technical} to
conclude that $Q_{1,V'}^{-+} \in \gS_2$. Obviously, the same holds true  
for $Q_{1,V'}^{+-}$, so that $Q_{1,V'}$, and henceforth $Q_{1,V}$, are in
$\gS_2$. As $|\nabla|(z-H^0_{\rm per})^{-1}$ is a bounded operator,
uniformly in $z \in \curv$, it is easy to check that $|\nabla|Q_{1,V_2}$
and $|\nabla|Q_{1,V'}$ both are Hilbert-Schmidt operators. Finally,
$|\nabla| Q_{1,V} \in \gS_2$ and therefore $Q_{1,V} \in \cQ$. As
$Q_{1,V}^{++}=Q_{1,V}^{--}=0$, we obviously get $\tr_0(Q_{1,V})=0$.

Let us now consider $Q_{k,V}$ for $k \ge 2$. The potential
$V$ being in $L^2(\R^3) + L^6(\R^3)$, we have 
\begin{itemize}
\item $Q_{2,V} \in \gS_3$ and $|\nabla|Q_{2,V} \in  \gS_3$;
\item $Q_{k,V} \in \gS_2$ and $|\nabla|Q_{k,V} \in \gS_2$ for all $k \ge
3$.
\end{itemize}
As usual 
\cite{HaiLewSer-05a,CanDelLew-08a}, the next step consists in
introducing $\gamma^0_{\rm per}+(\gamma^0_{\rm per})^\perp=1$ in
(\ref{eq:QkV}) in places where $(H^0_{\rm per}-z)^{-1}$ appears, and in
expanding everything. We will use the notation
$$
Q_{2,V}^{--+}:=-\frac1{2i\pi}\oint_\curv\frac{\gamma^0_{\rm per}}{H^0_{\rm
    per}-z}V\frac{\gamma^0_{\rm per}}{H^0_{\rm
    per}-z}V\frac{(\gamma^0_{\rm per})^\perp}{H^0_{\rm per}-z}dz,
$$
and similar definitions for all the other terms. A simple application of
the residuum formula tells us that  
$Q_{2,V}^{+++}=Q_{2,V}^{---}=0$. Therefore, $Q_{2,V}^{--} =
Q_{2,V}^{-+-}$ and  $Q_{2,V}^{++} = Q_{2,V}^{+-+}$.
Now we remark that the terms $Q_{2,V}^{-+-}$ and $Q_{2,V}^{+-+}$ involve two
terms of the form $\gamma^0_{\rm per}V(\gamma^0_{\rm
  per})^\perp=[\gamma^0_{\rm per},V](\gamma^0_{\rm per})^\perp$ (or its
adjoint) in their formula. Using Lemma~\ref{lem:technical}, we obtain
that $Q_{2,V}^{--}$, $Q_{2,V}^{++}$, $|\nabla|Q_{2,V}^{--}|\nabla|$  and
$|\nabla|Q_{2,V}^{++}|\nabla|$ are trace-class operators. Likewise, 
$Q_{k,V}^{--}$, $Q_{k,V}^{++}$, $|\nabla|Q_{k,V}^{--}|\nabla|$  and
$|\nabla|Q_{k,V}^{++}|\nabla|$ are trace-class operators.
Lastly, $Q_{2,V}^{-+}=Q_{2,V}^{--+}+Q_{2,V}^{-++}$, both operators of
the right-hand side involving one term of the form $\gamma^0_{\rm
  per}V(\gamma^0_{\rm per})^\perp=[\gamma^0_{\rm per},V](\gamma^0_{\rm
  per})^\perp$. Consequently $Q_{2,V}^{-+}$ and $|\nabla|Q_{2,V}^{-+}$
are Hilbert-Schmidt. Repeating the same argument for $Q_{2,V}^{+-}$, we
obtain that $Q_{2,V}$ and $|\nabla|Q_{2,V}$ are
Hilbert-Schmidt. Therefore, all the operators $Q_{k,V}$ are in $\cQ$. As
$Q_V$ also is in $\cQ$,
 $\widetilde Q_{k,V} \in \cQ$ for all $k \ge
3$. It also follows from the first assertion of
Lemma~\ref{lem:technical} and the Kato-Seiler-Simon inequality that 
$\widetilde Q_{k,V}$ is trace-class for $k \ge 6$.

Using (\ref{eq:commutator1}) and the Kato-Seiler-Simon inequality \eqref{KSS},
we then easily obtain that for all $k \ge 1$, there exists a constant
$C_k \in \R_+$ such that
$$
\forall V \in L^2(\R^3)+\cC', \quad \|Q_{k,V}\|_\cQ \le C_k \|V\|_{L^2+\cC'}^k
$$
and that for all $k \ge 2$, there exists a constant $\widetilde C_k \in
\R_+$ such that
$$
\forall V \in L^2(\R^3)+\cC' \mbox{ s.t. } \|V\|_{L^2+\cC'} < \alpha,
\quad \|\widetilde Q_{k,V}\|_\cQ \le \widetilde C_k \|V\|_{L^2+\cC'}^k.
$$
Let $V \in L^2(\R^3)+\cC'$ such that $\|V\|_{L^2+\cC'} < \alpha$.
For all $t \in [0,1]$, $\|tV\|_{L^2+\cC'} < \alpha$ and 
$$
Q_{tV} = Q_{1,tV} + \cdots + Q_{K,tV} + \widetilde Q_{K,tV} = 
t Q_{1,V} + \cdots + t^K Q_{K,tV} + \widetilde Q_{K+1,tV}.
$$
As we know that $\tr_0(Q_{tV}) = 0$, we obtain that for all $t \in [0,1]$,
$$
0 = t \tr_0(Q_{1,V}) + \cdots + t^K \tr_0(Q_{K,tV}) + \tr_0(\widetilde
Q_{K+1,tV}) 
$$
with $|\tr_0(\widetilde Q_{K+1,tV})| \le \|\widetilde Q_{K+1,tV}\|_\cQ
\le \widetilde C_{K+1} t^{K+1} \|V\|_{L^2+\cC'}^{K+1}$. Hence,
$\tr_0(Q_{k,V})=0$ for all $k \ge 1$ and $\tr_0(\widetilde Q_{k,V})=0$
for all $k \ge 2$. 
\end{proof}

\subsection{Proof of Proposition~\ref{prop:chi0}}\label{sec:proof_chi0}
Let $V \in L^2(\R^3)+\cC'$. We already know from Lemma~\ref{prop:Dyson} that 
$Q_{1,V} \in \cQ$ and that $\tr_0(Q_{1,V})=0$.
Decomposing $V$ as $V =
V_2 + V'$ with $V_2 \in L^2(\R^3)$ and $V'
\in \cC'$, and proceeding as in Section~\ref{sec:expansion_Q_V}, we
obtain 
$$\| Q_{1,V_2} \|_{\cQ}  \le  C \|V_2\|_{L^2}, $$
$$\| Q_{1,V'} \|_{\cQ}  \le  C \|V'\|_{\cC'} .$$
We infer that $\chi_0$ is a continuous linear application from
from $L^2(\R^3) + \cC'$ to $L^2(\R^3) \cap \cC$.

Let us now examine the case when $V \in L^1(\R^3)$. Using again the
Kato-Seiler-Simon inequality and 
the first assertion of Lemma~\ref{lem:technical}, we obtain
$Q_{1,V} \in \gS_1$ and 
$$
\| Q_{1,V} \|_{\gS_1} \le  C \|V\|_{L^1}.
$$
Consequently, $\chi_0$ defines a continuous linear application from
$L^1(\R^3)$ to $L^1(\R^3)$.
As from the residuum formula, $Q_{1,V}^{++}=Q_{1,V}^{--}=0$, we get
$\tr(Q_{1,V})=0$. 
 \qed

\subsection{Proof of Proposition~\ref{lem:limL}}
As $\rho \in L^1(\R^3) \subset v_{\rm c}^{-1}(L^2(\R^3)) + \cC$, we have
$\cL(\rho)\in L^2(\R^3) \cap \cC$. 
The operator $\cL$ can be explicitely calculated in Bloch transform. 
We start from the Bloch-Floquet decomposition of $H^0_{\rm per}$: for $f
\in H^2(\R^3)$,
$$
(H^0_{\rm per}f)(x) = \fint_{\Gamma^\ast} ((H^0_{\rm per})_q f_q) \,
e^{iq\cdot x} \, dq 
$$
where (see Eq.~\eqref{eq:dec_HOperq})
$$
(H^0_{\rm per})_q = \sum_{n=1}^{+\infty} \epsilon_{n,q}
|u_{n,q}\rangle\langle u_{n,q}|.
$$
Note that by time-reversal symmetry, 
$$
u_{n,-q}=\overline{u_{n,q}},\qquad \epsilon_{n,-q}=\epsilon_{n,q}.
$$
Denoting by $V = v_{\rm c}(\rho)$, we can write the Bloch matrix of the
operator $Q_{1,V}$ as: 
\begin{equation}
[Q_{1,V}]_{qq'}=\frac1{2i\pi}\oint_\curv(z-(H^0_{\rm
  per})_q)^{-1}V_{q-q'}(z-(H^0_{\rm per})_{q'})^{-1} \, dz. 
\label{matrix_Q_1}
\end{equation}
Inserting the spectral decomposition of $H^0_{\rm per}$ in \eqref{matrix_Q_1}, we obtain
\begin{multline}
[Q_{1,V}]_{qq'} = -
\sum_{n=1}^N\sum_{n'=N+1}^{+\ii}\bigg(\frac1{\epsilon_{n',q'}-\epsilon_{n,q}}
\pscal{u_{n,q},V_{q-q'}u_{n',q'}}_{L^2_{\rm
    per}(\Gamma)}|u_{n,q}\rangle\langle u_{n',q'}|\\
+\frac1{\epsilon_{n',q}-\epsilon_{n,q'}}\pscal{u_{n',q},V_{q-q'}u_{n,q'}}_{L^2_{\rm
    per}(\Gamma)}|u_{n',q}\rangle\langle u_{n,q'}|\bigg).\label{series_first}
\end{multline}
\begin{remark}
In the following we will write series of the form \eqref{series_first}
and we will invert sums and integrals without giving any
justification. To see that such a series is absolutely convergent, one
can use the fact that there exists $a$ and $b$ in $\R_+$ such that
for all $n \ge 1$ and $q \in \Gamma^\ast$,
$$
\epsilon_{n,q}\geq an^{2/3}-b.
$$
This bound is easily obtained by comparison with the eigenvalues of the
periodic Laplacian. It follows that there exists $C \in \RR_+$ such that 
$$
|\pscal{u_{n,q},V_{q-q'}u_{n',q'}}_{L^2_{\rm per}(\Gamma)}|\leq
\frac{C}{(n')^{2/3}}
$$ 
for all $1 \le n \le N$, all $n' \ge N+1$ and all $q,q' \in \Gamma^\ast$.
\end{remark}

If an operator $A \in \cQ$ has a Bloch matrix $A_{qq'}$, then we have
\begin{equation}
(\rho_A)_q(x)=\fint_{\Gamma^\ast} A_{q',q'-q}(x,x)\, dq'.
\label{rho_Bloch} 
\end{equation}
This formula is obtained by writing, for any real-valued function $f \in
L^2(\R^3)$, 
\begin{eqnarray*}
\fint_{\Gamma^\ast}\pscal{(\rho_A)_q,f_q}_{L^2_{\rm per}(\Gamma)}\,dq & =
& \int_{\R^3}\rho_A(x)f(x)\,dx\\ 
 & = & \tr(Af)= \fint_{\Gamma^\ast}\tr_{{L^2_{\rm per}(\Gamma)}}\big((Af)_{q',q'}\big)\, dq' \\
 & = & \fint_{\Gamma^\ast}dq' \fint_{\Gamma^\ast}dq''\; \tr_{L^2_{\rm per}(\Gamma)}\big(A_{q',q''}f_{q''-q'}\big)\\
 & = & \int_{\Gamma}dx\fint_{\Gamma^\ast}dq'\fint_{\Gamma^\ast}dq\; \overline{A_{q'+q,q'}(x,x)}{f_{q}(x)}.
\end{eqnarray*}
We deduce that 
\begin{multline} \label{eq:FB_L}
[\cL(\rho)]_{q}(x)=\fint_{\Gamma^\ast}dq'\sum_{n=1}^N\sum_{n'=N+1}^{+\ii} \\
\bigg(\frac1{\epsilon_{n',q'-q}-\epsilon_{n,q'}}\pscal{u_{n,q'},(\rho\star|\cdot|^{-1})_{q}
  u_{n',q'-q}}_{L^2_{\rm per}(\Gamma)}u_{n,q'}(x)\overline{u_{n',q'-q}(x)}\\ 
+\frac1{\epsilon_{n',q'}-\epsilon_{n,q'-q}}\pscal{u_{n',q'},(\rho\star|\cdot|^{-1})_{q}u_{n,q'-q}}_{L^2_{\rm per}(\Gamma)}u_{n',q'}(x)\overline{u_{n,q'-q}}(x)\bigg).
\end{multline}
The next step consists in decomposing the operator $\cL$ as the sum of a
singular part and a regular part, corresponding respectively to the
low and high Fourier modes of the Coulombic interaction kernel
$|\cdot|^{-1}$. More precisely,
we choose some smooth function $\xi$ which equals $1$ in a small
neighborhood $B(0,\delta)$ of 0 and 0 outside the ball $B(0,2\delta)$, with
$\delta >0$ such that $B(0,2\delta)\subset \Gamma^\ast$. Then we define 
$$
\cL_s(\rho):=\cL\left(\cF^{-1}\xi\widehat{\rho}\right),\qquad
\cL_r(\rho):=\cL\left(\cF^{-1}(1-\xi)\widehat{\rho}\right)
$$
where $\cF^{-1}$ is the inverse Fourier tranform. Similarly we define
$$
v_s:=\sqrt{\frac 2\pi}\cF^{-1}(\xi(\cdot)|\cdot|^{-2}),\qquad
v_r:=\sqrt{\frac 2\pi}\cF^{-1}((1-\xi(\cdot))|\cdot|^{-2})
$$
and note that $v_r\in L^1(\R^3)$. This being said, we have by the
Kato-Seiler-Simon inequality \eqref{KSS} that  
$$
Q_1^r:=\frac1{2i\pi}\oint_\curv(z-H^0_{\rm per})^{-1} (\rho\star v_r)
(z-H^0_{\rm per})^{-1}dz\in\gS_1,
$$
hence $\cL_r(\rho)\in L^1(\R^3)$ and $\int_{\R^3}\cL_r(\rho)=0$ by Proposition \ref{prop:chi0}. Consequently, 
$\widehat{\cL_r(\rho)} \in C^0(\RR^3)$ and $\widehat{\cL_r(\rho)}(0)=0$.

Let us now deal with the singular part of $\cL(\rho)$. 
Using the definition \eqref{def_Bloch_Floquet} of the Bloch-Floquet
transform, we obtain that
$$
(|\cdot|^{-1})_q(x)= \frac{4\pi}{|\Gamma|}
\sum_{K\in\cR^\ast}\frac{e^{iK\cdot x}}{|q+K|^2}
$$
and that for almost all $q \in \Gamma^\ast$,
$$
\left(\cF^{-1}(\xi\widehat{\rho})\right)_q(x)
=\frac{(2\pi)^{\frac 32}}{|\Gamma|}\widehat{\rho}(q)\xi(q).
$$
This implies that  for almost all $q \in \Gamma^\ast$,
\begin{equation}
(\rho \star v_s)_q(x)= 4\pi \frac{(2\pi)^{\frac 32}}{|\Gamma|}
\frac{\xi(q)\widehat{\rho}(q)}{|q|^2}. 
\label{formula_rho_vs} 
\end{equation}
Therefore we get for almost all $q \in \Gamma^\ast$,
\begin{equation}
\cL_s(\rho)_q(x)= (2\pi)^{\frac 32} \frac{B_q(x)}{|q|^2}\widehat{\rho}(q)
\label{def:L_singular} 
\end{equation}
where
\begin{multline}
 B_q(x):= \frac{4\pi}{|\Gamma|} 
 \xi(q)\fint_{\Gamma^\ast}d{q'}\sum_{n=1}^N\sum_{n'=N+1}^{+\ii} \\
\bigg(\frac1{\epsilon_{n',{q'}-q}-\epsilon_{n,{q'}}}\pscal{u_{n,{q'}},u_{n',{q'}-q}}_{L^2_{\rm per}(\Gamma)}u_{n,{q'}}(x)\overline{u_{n',{q'}-q}(x)}\\
+\frac1{\epsilon_{n',{q'}}-\epsilon_{n,{q'}-q}}\pscal{u_{n',{q'}},u_{n,{q'}-q}}_{L^2_{\rm per}(\Gamma)}u_{n',{q'}}(x)\overline{u_{n,{q'}-q}(x)}\bigg).
\label{def_B}
\end{multline}
It follows that almost everywhere in $\Gamma^\ast$,
\begin{eqnarray*}
\widehat{\cL_s(\rho)}(q) & = & (2\pi)^{-\frac 32} 
\int_{\Gamma} \cL_s(\rho)_q(x)  dx \\
& = & \frac{{\mathcal B}(q)}{|q|^2}\widehat{\rho}(q)
\end{eqnarray*}
where
\begin{equation}
{\mathcal B}(q) = \frac{8\pi}{|\Gamma|}  \xi(q)
\sum_{n=1}^N\sum_{n'=N+1}^{+\ii} \fint_{\Gamma^\ast}d{q'}
\frac{\left|\pscal{u_{n,{q'}},u_{n',{q'}-q}}_{L^2_{\rm
        per}(\Gamma)}\right|^2}{\epsilon_{n',{q'-q}}-\epsilon_{n,{q'}}}. 
\label{def_cB_q}
\end{equation}
We now remark that the above formula may be written
$$
{\mathcal B}(q) = -\frac{8\pi}{|\Gamma|}  \xi(q)
\tr_{L^2_{\rm per}}\left[ \oint_\curv dz\fint_{\Gamma^\ast}d{q'}\frac{\left(\gamma^0_{\rm per}\right)_{q'}}{z-(H^0_{\rm per})_{q'}}\;\frac{\left(\gamma^0_{\rm per}\right)_{q'-q}^\perp}{z-(H^0_{\rm per})_{q'-q}}\right]. 
$$
We recall \cite{Panati-07} that $q\mapsto\left(\gamma^0_{\rm
    per}\right)_q$ is a smooth periodic function and that
$\left(\gamma^0_{\rm per}\right)_q$ is for all $q$ a rank-$N$ orthogonal
projector. It is then easy to deduce that $q\mapsto\cB(q)$ is a
continuous periodic function on $\R^3$. 
Consequently, $\widehat{\cL_s(\rho)}$ and therefore
$\widehat{\cL(\rho)}$ are continuous on
$\Gamma^\ast \setminus \left\{0\right\}$. Using similar arguments, one
obtains that $\widetilde{\cL(\rho)}$ is continuous on $\R^3 \setminus
\cR^\ast$. 
In order to study the limit of $\widehat{\cL_s(\rho)}(q)$ when $q$ goes
to zero, we use the relation
\begin{multline}
\left[ \left( \epsilon_{n',q'-q}-\frac{|q'-q|^2}2 \right) - 
\left( \epsilon_{n,q'}-\frac{|q'|^2}2 \right) \right] \langle
u_{n,q'},u_{n'q'-q}\rangle_{L^2_{\rm per}(\Gamma)}  \\
= -
\langle i q \cdot \nabla u_{n,q'},u_{n',q'-q} \rangle_{L^2_{\rm per}(\Gamma)}
\label{pscal_u_n_q}
\end{multline}
to rewrite $\widehat{\cL_s(\rho)}(\eta \sigma)$ for $\sigma \in S^2$ and
$\eta>0$ small enough as 
$$
\widehat{\cL_s(\rho)}(\eta \sigma) = L_\eta(\sigma) 
\widehat{\rho}(\eta \sigma)
$$
where
$$
L_\eta(\sigma) =  \frac{8\pi}{|\Gamma|} 
\sum_{n=1}^N\sum_{n'=N+1}^{+\ii} \fint_{\Gamma^\ast}d{q'}
\frac{\left|
\langle \sigma \cdot \nabla u_{n,q'},u_{n',q'-\eta \sigma} 
\rangle_{L^2_{\rm per}(\Gamma)}\right|^2}
{(\epsilon_{n',{q'-\eta\sigma}}-\epsilon_{n,{q'}})
(\epsilon_{n',{q'-\eta\sigma}}-\epsilon_{n,{q'}}+\eta q'\cdot
\sigma-\frac{\eta^2}2)^2} 
$$
(recall that $\xi \equiv 1$ in the vicinity of $0$). 
Again the above formula may be rewritten as
\begin{multline*}
L_\eta(\sigma) =  \frac{8\pi}{|\Gamma|} 
\sum_{n=1}^N \fint_{\Gamma^\ast}d{q'}\times\\
\pscal{\frac{\left(\gamma^0_{\rm per}\right)^\perp_{q'-\eta\sigma}}
{\left((H^0_{\rm per})_{q'-\eta\sigma}-\epsilon_{n,{q'}}\right)
\left((H^0_{\rm per})_{q'-\eta\sigma}-\epsilon_{n,{q'}}+\eta q'\cdot
\sigma-\frac{\eta^2}2\right)^2}\sigma \cdot \nabla u_{n,q'},\sigma \cdot \nabla u_{n,q'}}_{L^2_{\rm per}(\Gamma)} 
\end{multline*}
which shows that when $\eta$ goes
to zero, $L_\eta(\sigma)$ converges to $\sigma^TL\sigma$ while
$\widehat{\rho}(\eta \sigma)$ converges to $\widehat{\rho}(0)$. 

\medskip

We now turn to the proof that $L_0>0$. We note first that
$$L_0= \frac13 \frac{8\pi}{|\Gamma|} 
\sum_{n=1}^N\sum_{n'=N+1}^{+\ii} \fint_{\Gamma^\ast}d{q'}
\frac{\left|
\langle \nabla u_{n,q'},u_{n',q'} 
\rangle_{L^2_{\rm per}(\Gamma)}\right|^2}
{(\epsilon_{n',{q'}}-\epsilon_{n,{q'}})^3}$$
hence
$L_0=0$ would imply $\langle \nabla u_{n,q},u_{n',q} \rangle_{L^2_{\rm per}(\Gamma)}=0$ for all $q\in\Gamma^\ast$, all $n\leq N$ and all $n'\geq N+1$. 
Hence for $i=1,2,3$, $\partial_{x_i}$ would stabilize the space $X_q$ spanned by $(u_{1,q},...,u_{N,q})$ for any $q$.
Next we differentiate the eigenvalue equation for $u_{n,q}$ and get
$$\left(-\frac{\Delta}{2}-iq\cdot\nabla+\frac{|q|^2}{2}+V_{\rm per}-\epsilon_{n,q}\right)\partial_{x_i}u_{n,q}+(\partial_{x_i}V_{\rm per})u_{n,q}=0.$$
From this we deduce that $(\partial_{x_i}V_{\rm per})$ would also
stabilize $X_q$. This means that we would have
$$
\forall x\in\Gamma,\qquad (\partial_{x_i}V_{\rm per})(x)\left(\begin{array}{c}u_{1,q}(x)\\ \vdots\\ u_{N,q}(x)\\ \end{array}\right)
=A_q\left(\begin{array}{c}u_{1,q}(x)\\ \vdots\\ u_{N,q}(x)\\
  \end{array}\right)
$$
for some matrix $A_q$ depending only on $q$. As $u_{1,0}(x)>0$ for all
$x$ (it is the first eigenfunction of a Schrödinger operator), we deduce
that $(\partial_{x_i}V_{\rm per})(x)$ would be for any $x\in\Gamma$ an
eigenvalue of $A_0$. By continuity and periodicity we infer that $V_{\rm
  per}$ would be constant. This is in contradiction with the assumption
that the host crystal is an insulator or a semiconductor.
\qed

\subsection{Proof of Proposition~\ref{lem:defL}}
The proof of Proposition \ref{prop:chi0} shows that $\cL$
defines a bounded linear operator on $\cC$. Besides, for all $\rho_1$
and $\rho_2$ in $\cC$, 
\begin{eqnarray*}
\langle \cL(\rho_1),\rho_2 \rangle & = & \int_{\R^3} [\cL(\rho_1)](x) \,
[v_{\rm c}(\rho_2)](x) \, dx \\
& = & - \frac1{2i\pi}\oint_\curv \tr \left( (z-H^0_{\rm
  per})^{-1} v_{\rm c}(\rho_1) (z-H^0_{\rm per})^{-1} v_{\rm c}(\rho_2)
\right) \, dz \\
& = & \langle \cL(\rho_2),\rho_1 \rangle = \langle \rho_1, \cL(\rho_2) \rangle.
\end{eqnarray*}
Therefore, $\cL$ is self-adjoint on $\cC$.
Lastly, denoting by $V = v_{\rm c}(\rho)$, we have for all $\rho \in \cC$,
\begin{align*}
\pscal{\cL(\rho),\rho}_{\cC}&=\int_{\R^3}\cL(\rho)(x)V(x)\,dx=\fint_{\Gamma^\ast}\int_{\Gamma}
\overline{\cL(\rho)_q(x)} \, V_q(x)\,dx\,dq\\ 
&=\sum_{n=1}^N\sum_{n'=N+1}^{+\ii}\fint_{\Gamma^\ast}dq\fint_{\Gamma^\ast}d{q'}\Bigg(\frac{\big|\pscal{u_{n,{q'}},V_{q}u_{n',{q'}-q}}_{L^2_{\rm
      per}(\Gamma)}\big|^2}{\epsilon_{n',{q'}-q}-\epsilon_{n,{q'}}}\\
&\qquad\qquad\qquad\qquad\qquad+\frac{\left|\pscal{u_{n',{q'}},V_{q}u_{n,{q'}-q}}_{L^2_{\rm per}(\Gamma)}\right|^2}{\epsilon_{n',{q'}}-\epsilon_{n,{q'}-q}}\Bigg)\geq0.
\end{align*}
We conclude that $\cL$ is a bounded positive self-adjoint operator on
$\cC$. Consequently, $1+\cL \, : \, \cC \rightarrow \cC$ is invertible.
 \qed

\subsection{Expanding the density to higher orders}\label{sec:expansion_density_higher_orders}
In the previous sections, we have studied the first order density
$-\cL(\rho)$. For the proof of our Theorem \ref{thm:notL1} on the
reduced-Hartree-Fock model, we need to consider the higher order terms.
Each of the operators $Q_{k,V}$ (for $k \ge 1$) and $\widetilde
Q_{k,V}$ (for $k \ge 2$) defined in Lemma \ref{prop:Bogoliubov} being in $\cQ$, the expansion
(\ref{eq:Dyson_Q}) can be rewritten in terms of the associated
densities, yielding the following equation in $L^2(\R^3) \cap \cC$:
\begin{equation} \label{eq:Dyson_rho}
\rho_{Q_V} =   \rho_{Q_{1,V}} + \cdots + \rho_{Q_{K,V}} +
\rho_{\widetilde Q_{K+1,V}}.  
\end{equation}
We introduce the quadratic operator $r_2$  defined by
$$
r_2(\rho) = \rho_{Q_{2,v_{\rm c}(\rho)}},
$$
which is continuous from $v_{\rm c}^{-1}(L^2(\R^3))+\cC$ to $L^2(\R^3)
\cap \cC$, and the nonlinear map $\widetilde r_3$ from
$v_{\rm c}^{-1}(B_\alpha)$ to $L^2(\R^3) \cap \cC$ ($B_\alpha$ denoting
the ball of $L^2(\R^3)+\cC'$ of radius $\alpha$), defined by
$$
\widetilde r_3(\rho) = \rho_{\widetilde Q_{3,v_{\rm c}(\rho)}}.
$$
We obtain
\begin{equation}
r(\rho):=\rho_{Q_{v_{\rm c}(\rho)}}  = - \cL(\rho) + r_2(\rho)  + 
\widetilde r_3(\rho).\label{eq:Dyson_rho_2} 
\end{equation}
The next lemma is concerned with the second and third order terms of the
expansion (\ref{eq:Dyson_rho_2}). We will assume
that $\rho \in L^1(\R^3) \cap L^2(\R^3)$. Then $\rho \in L^{\frac
  65}(\R^3) \subset \cC$, so that $v_{\rm c}(\rho) \in \cC'$ and
$\|v_{\rm c}(\rho)\|_{\cC'} \le C 
\|\rho\|_{L^1 \cap L^2}$, where $C$ is a universal constant. In
particular, there exists a constant $\gamma > 0$, such that 
$$
\| \rho\|_{L^1 \cap L^2} \le \gamma \quad \Rightarrow \quad \|v_{\rm
  c}(\rho)\|_{L^2 + \cC'} \le \|v_{\rm c}(\rho)\|_{\cC'} < \alpha. 
$$ 

\begin{lemma}[Nonlinear terms in the expansion]  \label{lem:order2}
Let $\rho \in L^1(\RR^3) \cap L^2(\R^3)$. Then
\begin{enumerate}
\item $Q_{2}(\rho) \in \gS_p$ for all $p > 1$ and the 
Fourier transform $\widehat{r_2(\rho)}$ of $r_2(\rho)$ is continuous on
$\R^3$ and vanishes at $k=0$~;
\item If in addition, $\|\rho\|_{L^1 \cap L^2} \le \gamma$, then
  $\widetilde Q_3(\rho) \in \gS_1$,  $\widetilde r_3(\rho) \in L^1(\R^3)$
  and 
$$
\tr(\widetilde Q_{3}(\rho)) = \int_{\RR^3} \widetilde r_3(\rho) = 0.
$$ 
\end{enumerate}
\end{lemma}

\begin{proof}[Proof of Lemma~\ref{lem:order2}]
As $\rho \in L^1(\R^3) \cap L^2(\R^3)$, we deduce from Young
inequality that $V = \rho \star |\cdot|^{-1}$ is in $L^p(\R^3)$ for $3
< p < \infty$ and that $\nabla V$ is in $(L^q(\R^3))^3$ for all $\frac 32 <
q < 6$. Therefore,  
$$
\left\|[\gamma^0_{\rm per},V]\right\|_{\gS_q}\leq C_q \|\nabla V\|_{L^q}
$$
for all $\frac 32 < q < 6$, by Lemma \ref{lem:technical}. Arguing as in the proof of
Lemma~\ref{prop:Dyson}, we obtain that $Q_{2,V} \in \gS_p$ for all
$p > 1$. 

We now concentrate on the regularity of $\widehat{r_2(\rho)}=\widehat\rho_{Q_{2,V}}$ at the origin. Arguing like in the proof of Proposition \ref{lem:limL}, we only have to study the density associated with the operator
$$R_2:=\frac1{2i\pi}\oint_\curv(z-H^0_{\rm per})^{-1}\left(\rho\star v_s\right)(z-H^0_{\rm per})^{-1}\left(\rho\star v_s\right)(z-H^0_{\rm per})^{-1}dz.$$
We will for simplicity only treat the term $R_2^{-++}$, the other ones being similar. Following the proof of Proposition \ref{lem:limL}, we obtain
\begin{multline}
\widehat{\rho_{R_2^{-++}}}(q)=(2\pi)^{-\frac{3}{2}}\sum_{n=1}^N\sum_{m\geq N+1}\sum_{m'\geq N+1}\fint_{\Gamma^\ast}dq'\fint_{\Gamma^\ast}dr\\
\frac{\pscal{u_{n,q'},(\rho\star v_s)_{q'-r}u_{m,r}}\pscal{u_{m,r},(\rho\star v_s)_{r-q'+q}u_{m',q'-q}}\pscal{u_{m',q'-q},u_{n,q'}}}{(\epsilon_{n,q'}-\epsilon_{m,r})(\epsilon_{n,q'}-\epsilon_{m',q'-q})}.
\end{multline}
Changing $r\leftarrow r-q'$ and using as before \eqref{formula_rho_vs},
we see that for $|q|$ small enough,
\begin{multline}
\widehat{\rho_{R_2^{-++}}}(q)=4\sqrt{2\pi}\sum_{n=1}^N\sum_{m\geq N+1}\sum_{m'\geq N+1}\fint_{\Gamma^\ast}dq'\fint_{\Gamma^\ast}dr\xi(-r) \xi(r+q)\times\\
\times\frac{\widehat{\rho}(-r) \widehat{\rho}(r+q)\pscal{u_{n,q'},u_{m,r+q'}}\pscal{u_{m,r+q'},u_{m',q'-q}}\pscal{u_{m',q'-q},u_{n,q'}}}{|r|^2|r+q|^2(\epsilon_{n,q'}-\epsilon_{m,r-q'})(\epsilon_{n,q'}-\epsilon_{m',q'-q})}.
\end{multline}
Next, using \eqref{pscal_u_n_q}, we obtain
\begin{multline}
|\widehat{\rho_{R_2^{-++}}}(q)|\leq C|q|\sum_{n=1}^N\sum_{m\geq N+1}\sum_{m'\geq N+1} \fint_{\Gamma^\ast}dq'\fint_{\Gamma^\ast}dr\frac{\xi(-r) \xi(r+q)}{|r|\;|r+q|^2}\times\\
\times\frac{|\widehat{\rho}(-r)|\; |\pscal{\nabla u_{n,q'}, u_{m,r+q'}}|}{|\epsilon_{m,r+q'}-\epsilon_{n,q'}-r\cdot(r+2q')/2|}\times\\
\times\frac{|\widehat{\rho}(r+q)|\;  |\pscal{u_{m',q'-q},\nabla u_{n,q'}}|}{(\epsilon_{m,r-q'}-\epsilon_{n,q'})(\epsilon_{m',q'-q}-\epsilon_{n,q'})|\epsilon_{n,q'}-\epsilon_{m',q'-q}+q\cdot(q-2q')/2|}.
\end{multline}
Note that choosing the support of $\xi$ small enough we have
$$|\epsilon_{m,r+q'}-\epsilon_{n,q'}-r\cdot(r+2q')/2|\geq c(m^{2/3}+1)$$
uniformly for $r,q'\in\Gamma^{\ast}$ and $n=1..N$. Similarly, taking $q$ small enough we get
$$|\epsilon_{n,q'}-\epsilon_{m',q'-q}+q\cdot(q-2q')/2|\geq c((m')^{2/3}+1).$$
Using these estimates and the fact that $\rho\in L^1$  we deduce that
$$|\widehat{\rho_{R_2^{-++}}}(q)|\leq C|q|\int_{\R^3}\frac{\xi(-r)}{|r|\;
  |r+q|^2} \, dr \leq C|q|\log\frac1{|q|}$$
and the result follows.

To establish that $\widetilde Q_{3,V}$ is trace-class, and therefore
that $\widetilde r_3(\rho)$ is integrable, we write 
$$
\widetilde Q_{3,V} = Q_{3,V}+\widetilde Q_{4,V}
$$
and proceed as above to prove that both operators in the right hand side
are trace-class. As $\tr_0(\widetilde Q_{3,V})=0$, we readily conclude
that $\tr(\widetilde Q_{3,V})=\int_{\R^3} \widetilde r_3(\rho)= 0$.
\end{proof}

\subsection{Proof of Theorem \ref{thm:notL1}}\label{sec:proof_notL1}
We now have all the material for proving Theorem \ref{thm:notL1}.
The first step is to confirm that if the external potential $ \nu \star |\cdot|^{-1}$
is small, so is the
effective potential $ (\nu -\rho_{\nu,\epsilon_{\rm F}}) \star
|\cdot|^{-1}$, hence the results of Section \ref{sec:expansion_Q_V} can be applied. 

\begin{lemma} \label{prop:small}
There exists $\beta > 0$ such that if
$$
\| \nu \star |\cdot|^{-1}\|_{L^2+\cC'} < \beta,
$$
then 
$$
\| (\nu -\rho_{\nu,\epsilon_{\rm F}}) \star |\cdot|^{-1}\|_{L^2+\cC'}
< \alpha
$$
where $\alpha$ is the constant used in the formulation of
Lemma~\ref{prop:Bogoliubov}. Consequently,  
the solution to (\ref{eq:min_E}) is unique and satisfies
$\tr_0(Q_{\nu,\epsilon_{\rm F}}) = 0$ and
\begin{eqnarray} 
Q_{\nu,\epsilon_{\rm F}} & = & 
1_{(-\infty,\epsilon_{\rm F}]} \left(H^0_{\rm per} + (\rho_{\nu,\epsilon_{\rm
    F}}-\nu) \star |\cdot|^{-1} \right) -  
1_{(-\infty,\epsilon_{\rm F}]} \left(H^0_{\rm per}\right) \nonumber \\
& = & \frac{1}{2i\pi} \oint_\curv \left[
\left( z-H^0_{\rm per} - (\rho_{\nu,\epsilon_{\rm
    F}}-\nu) \star |\cdot|^{-1} \right)^{-1} - \left(z-H^0_{\rm
  per}\right)^{-1} \right] \, dz.
\label{eq:SCFsmall}
\end{eqnarray}
\end{lemma}

\begin{proof}[Proof of Lemma~\ref{prop:small}]
Let $0 < \delta <1$ and $\nu$
such that $\|\nu \star |\cdot|^{-1}\|_{L^2 + \cC'} \le \delta$. This implies that 
$$
\nu \star |\cdot|^{-1} = V_2 + \nu' \star |\cdot|^{-1}
$$
with $V_2 \in L^2(\R^3)$, $\nu' \in \cC$, $\|V_2\|_{L^2} \le \delta$ and
$\|\nu'\|_\cC \le \delta$. We then deduce from (\ref{eq:bounds}) that
$$
0 =  E_{\nu,\epsilon_{\rm F}}(0) \ge E_{\nu,\epsilon_{\rm
    F}}(Q_{\nu,\epsilon_{\rm F}}) \ge \frac{c_1}2
\tr\left((1+|\nabla|^2)(Q^{++}_{\nu,\epsilon_{\rm F}}-
Q^{--}_{\nu,\epsilon_{\rm F}})\right) - C' \delta - \frac{\delta^2}2.
$$
It follows that there exists a constant $c \in \RR_+$ independent of
$\delta$ and $\nu$ such that 
$$
\tr\left((1+|\nabla|^2)(Q^{++}_{\nu,\epsilon_{\rm F}}-
Q^{--}_{\nu,\epsilon_{\rm F}})\right) \le c \delta.
$$
Using again the inequalities $Q^2 \le Q^{++}-Q^{--}$, $Q^{++} \ge 0$,
$Q^{--} \le 0$, we obtain
$$
\| Q_{\nu,\epsilon_{\rm F}} \|_{\cQ} \le 2c \delta^{\frac 12}.
$$
Therefore, there exists a constant $c'$ independent of $\delta$ such that
for all $\nu$ such that $\|\nu \star |\cdot|^{-1}\|_{L^2 + \cC'} \le
\delta$, 
$$
\| (\rho_{\nu,\epsilon_{\rm F}}-\nu) \star |\cdot|^{-1}\|_{L^2 + \cC'}
\le c' \delta^{\frac 12}.
$$
We obtain the desired result by choosing $\beta = \min
(1,{c'}^{-2}\alpha^2)$.  
\end{proof}

The proof of Theorem \ref{thm:notL1} is a simple consequence of the results of Sections \ref{sec:expansion_Q_V} and \ref{sec:expansion_density_higher_orders}. We assume that $\| \nu \star |\cdot|^{-1}\|_{L^2+\cC'} < \beta$ in such a way that Lemma \ref{prop:small} can be applied. This gives us that $\tr_0(Q_{\nu,\epsilon_{\rm F}})=0$ and that $\| (\nu -\rho_{\nu,\epsilon_{\rm F}}) \star |\cdot|^{-1}\|_{L^2+\cC'}
< \alpha$. Hence we can use the expansion of Lemma \ref{prop:Bogoliubov}.

If
$Q_{\nu,\epsilon_{\rm F}}$ were trace-class, then we would have 
$\rho_{{\nu,\epsilon_{\rm F}}} \in L^1(\R^3)$ and
$$
\widehat\rho_{{\nu,\epsilon_{\rm F}}}(0) = \int_{\R^3}
\rho_{{\nu,\epsilon_{\rm F}}} = \tr(Q_{\nu,\epsilon_{\rm F}}) = 0.
$$
On the other hand,
we would obtain from the expansion (\ref{eq:Dyson_rho_2}) 
\begin{equation} \label{eq:SCFrho}
\rho_{{\nu,\epsilon_{\rm F}}}   =  -  \cL (\rho_{{\nu,\epsilon_{\rm
      F}}}-\nu) + r_2(\rho_{{\nu,\epsilon_{\rm F}}}-\nu) +  
\widetilde r_3(\rho_{{\nu,\epsilon_{\rm F}}}-\nu).
\end{equation}
By Proposition \ref{lem:limL} and since we have assumed $\rho_{{\nu,\epsilon_{\rm F}}} \in L^1(\R^3)$, we know that 
$$\lim_{\eta\to0^+}\left(\widehat{\rho_{{\nu,\epsilon_{\rm F}}}}(\eta \sigma)+\widehat{\cL (\rho_{{\nu,\epsilon_{\rm
      F}}}-\nu)}(\eta\sigma)\right)=- (\sigma^TL\sigma)\widehat{\nu}(0)$$
for all $\sigma\in S^2$. On the other hand we have by Lemma \ref{lem:order2} that
the Fourier transform of the second and third order terms $r_2(\rho_{{\nu,\epsilon_{\rm F}}}-\nu)$ and 
$\widetilde r_3(\rho_{{\nu,\epsilon_{\rm F}}}-\nu)$ vanish at the origin. 
It would then follow that $(\sigma^TL\sigma)\widehat{\nu}(0) = 0$ for all $\sigma \in S^2$, which
obviously contradicts (\ref{eq:Lpositive}). Therefore,
$Q_{\nu,\epsilon_{\rm F}}$ is not trace-class.

Let us know assume that $\rho_{{\nu,\epsilon_{\rm F}}} \in
L^1(\R^3)$. The same arguments lead to 
$$
\widehat\rho_{{\nu,\epsilon_{\rm F}}}(0) = - (\sigma^TL\sigma)
(\widehat\rho_{{\nu,\epsilon_{\rm F}}}(0)-\widehat\nu(0))
$$
for all $\sigma \in S^2$. This is only possible if $L=L_0$. \qed

\subsection{Proof of Proposition~\ref{thm:dielectric}}\label{sec:proof_dielectric}
Let $\epsilon = 1-v_{\rm c} \chi_0$. It follows from
Proposition~\ref{prop:chi0} that $\epsilon$ is a bounded self-adjoint operator
on $\cC'$. Besides, using the fact that $\cL : = - \chi_0 v_{\rm c}$,
we easily see that 
\begin{eqnarray*}
\left[1+v_{\rm c}(1+\cL)^{-1} \chi_0 \right] \epsilon & = &
\left[1+v_{\rm c}(1+\cL)^{-1} \chi_0 \right] (1-v_{\rm c} \chi_0) \\
& = & 1 - v_{\rm c} \chi_0 + v_{\rm c} (1+\cL)^{-1} \chi_0 
+ v_{\rm c} (1+\cL)^{-1} \cL \chi_0 = 1.
\end{eqnarray*}
Likewise, $\epsilon\left[1+v_{\rm c}(1+\cL)^{-1} \chi_0 \right]=1$. Hence,
$
\epsilon^{-1} = \left[1+v_{\rm c}(1+\cL)^{-1} \chi_0 \right]^{-1}.
$

Lastly, $v_{\rm c}^{\frac{1}{2}}$ is an invertible bounded linear operator
from $L^2(\R^3)$ onto $\cC'$. Besides, for all $f$ and $g$ in
$L^2(\R^3)$, 
$$
\pscal{v_{\rm c}^{\frac 12}f,v_{\rm c}^{\frac 12}g}_{\cC'} 
= \langle f,g \rangle_{L^2}.
$$
As $\epsilon$ is an invertible bounded self-adjoint operator on $\cC'$,
$\widetilde \epsilon$ is an invertible bounded self-adjoint operator on
$L^2(\R^3)$:
$$\langle \widetilde \epsilon f,g \rangle_{L^2} = 
\pscal{v_{\rm c}^{\frac 12} \widetilde \epsilon 
f, v_{\rm c}^{\frac 12} g}_{\cC'} 
= \pscal{v_{\rm c}^{\frac 12} f, 
v_{\rm c}^{\frac 12} \widetilde \epsilon g }_{\cC'} 
 = \pscal{ f,  \widetilde \epsilon g }_{L^2}. $$
The proof is complete.
\qed

\subsection{Proof of Theorem~\ref{thm:homogenization}}

Let $\nu \in L^1(\R^3) \cap L^2(\R^3)$. Introducing the dilation
operator $(U_\eta f)(x)=\eta^{\frac 32}f(\eta x)$, we can write
$\nu_\eta = \eta^{\frac 32} U_\eta \nu$. The operator $U_\eta$ is an
isometry of $L^2(\R^3)$ and satisfies $\widetilde{U_\eta\phi}(k) =
\eta^{-\frac 32} \widetilde{\phi}(\eta^{-1}k)$. It follows that 
$$
\widehat{\nu_\eta}(k) = \widehat{\nu}\left( \frac k \eta \right) ,
$$
yielding
$$
\|\nu_\eta\|_{\cC} = \eta^{\frac{1}{2}} \|\nu\|_\cC.
$$
Hence, for $\eta > 0$ small enough, $\|\nu_\eta \star |\cdot|^{-1}
\|_{L^2+\cC'} < \beta$. 
Arguing like in the proof of Lemma~\ref{prop:small} we obtain
$$\frac12\|\rho_{\nu_\eta,\epsilon_{\rm F}}-\nu_\eta\|_\cC^2-\frac12\|\nu_\eta\|_\cC^2\leq E^{\nu_\eta}_{\epsilon_{\rm F}}(Q_{\nu_\eta,\epsilon_{\rm F}})\leq E^{\nu_\eta}_{\epsilon_{\rm F}}(0)=0.$$
Therefore,
$$\|\rho_{\nu_\eta,\epsilon_{\rm F}}-\nu_\eta\|_\cC\leq \eta^{\frac{1}{2}}\|\nu\|_\cC.$$
We therefore may use the self-consistent equation
\eqref{eq:SCF_rewritten2} and get 
\begin{equation}
\nu_\eta-\rho_{\nu_\eta,\epsilon_{\rm F}}=(1+\cL)^{-1}\nu_\eta-(1+\cL)^{-1}\tilde{r}_2(\nu_\eta-\rho_{\nu_\eta,\epsilon_{\rm F}}),
\label{eq:SCF_rewritten3}
\end{equation}
where $\tilde{r}_2(\rho):=r_2(\rho)+\tilde{r}_3(\rho)$.
The bounds of the proof of Proposition \ref{prop:Dyson} and the fact that $(1+\cL)^{-1}$ is a bounded operator on $\cC$ imply
\begin{equation}
\|(1+\cL)^{-1}\tilde{r}_2(\nu_\eta-\rho_{\nu_\eta,\epsilon_{\rm
    F}})\|_\cC\leq C \|\nu_\eta-\rho_{\nu_\eta,\epsilon_{\rm
    F}}\|_\cC^2\leq C'\eta.  
\label{estim_second_order_eta}
\end{equation}

For convenience, we are going to study equation~\eqref{eq:SCF_rewritten3} in
$L^2(\R^3)$. We therefore introduce 
$$
f_\eta:=v_c^{-\frac{1}{2}}W^\eta_\nu\qquad\text{and}\qquad g:=v_c^{\frac{1}{2}}\nu.
$$
We note that $f_\eta$ is indeed bounded in $L^2(\R^3)$ by the choice of
the scaling in $W^\eta_\nu$. Making use of the relation 
$\eta v_{\rm c}^{\frac 12}U_\eta v_{\rm c}^{-\frac 12} = U_\eta$, we can rewrite
\eqref{eq:SCF_rewritten3} as
\begin{equation}
f_\eta=U_\eta^\ast\,\tilde\epsilon^{-1}\, U_\eta
g-\eta^{-\frac{1}{2}}\,U_\eta^\ast v_c^{\frac{1}{2}}(1+\cL)^{-1}\tilde{r}_2(\nu_\eta-\rho_{\nu_\eta,\epsilon_{\rm F}}),
\label{eq:SCF_rewritten4}
\end{equation}
where we recall that $\tilde{\epsilon}^{-1}=v_c^{\frac{1}{2}}(1+\cL)^{-1}v_c^{-\frac{1}{2}}$ is a bounded self-adjoint operator on $L^2(\R^3)$. Our bound \eqref{eq:SCF_rewritten4} on the nonlinear term shows that
$$\|f_\eta-U_\eta^\ast\,\tilde\epsilon^{-1}\, U_\eta g\|_{L^2(\R^3)}\leq C\eta^{\frac{1}{2}}.$$
Hence the theorem will be proved if we show that
$U_\eta^\ast\,\tilde\epsilon^{-1}\, U_\eta g$ converges weakly in
$L^2(\R^3)$ to the correct limit as $\eta\to0^+$. 

This will follow from the following two important lemmas.
\begin{lemma}[The macroscopic dielectric permittivity]\label{lem:def_epsilon} We denote by $\tilde{\epsilon}_q:L^2_{\rm per}(\Gamma)\to L^2_{\rm per}(\Gamma)$ the Bloch transform of the operator $\tilde{\epsilon}$. We also denote by $e_0:=|\Gamma|^{-\frac{1}{2}}$ the (normalized) constant function and by $P_0$ the orthogonal projection on $\{e_0\}^\perp$. The following hold:
\begin{enumerate}
 \item The maps $q\mapsto \tilde{\epsilon}_q$ and $q\mapsto \tilde{\epsilon}^{-1}_q$ are continuous on $\Gamma^\ast\setminus\{0\}$ and uniformly bounded with respect to $q$.
\item For all $\sigma\in S^2$, $\tilde{\epsilon}_{t\sigma}e_0$ converges
  strongly in $L^2_{\rm per}(\Gamma)$ to  $b_\sigma(x)$,
  where for all $k \in \R^3$, the periodic function $b_k(x)$ is defined by
\begin{multline}
 b_k =(|k|^2+k^TLk)e_0  \\
-\frac{2i\sqrt{4\pi}}{|\Gamma|^{\frac 12}}
G_0^{\frac{1}{2}}\fint_{\Gamma^\ast}dq'\sum_{n=1}^N
\left(\frac{(\gamma_{\rm per}^0)^\perp_{q'}}{\left((H^0_{\rm per})_{q'}-\epsilon_{n,q'}\right)^2}(k\cdot\nabla)u_{n,q'}\right)\overline{u_{n,q'}},
\end{multline}
and where
$G^{\frac{1}{2}}_0$ is the operator defined on $L^2_{\rm per}(\Gamma)$
as
$$
G_0^{\frac{1}{2}}f
= \sum_{K\in\cR^\ast\setminus\{0\}} \frac{\sqrt{4\pi} \,
  \widehat{f}_K}{|K|} \, \frac{e^{iK\cdot x}}{|\Gamma|^{\frac 12}}
\qquad \mbox{where} \qquad \widehat{f}_K = \int_\Gamma f(x) \frac{e^{-iK\cdot
    x}}{|\Gamma|^{\frac 12}} \, dx,
$$
and which satisfies $P_0G_0^{\frac{1}{2}}=G_0^{\frac{1}{2}}P_0$.
\item The family of operators $P_0\tilde\epsilon_q P_0$ seen as bounded self-adjoint operators acting on $P_0L^2_{\rm per}(\Gamma)$ is continuous with respect to $q$ and one has
$$\big(P_0\tilde\epsilon_q P_0\big)_{|P_0L^2_{\rm per}(\Gamma)}\to C$$
strongly as $q\to0$, where $C\geq 1$ is the bounded operator on
$P_0L^2_{\rm per}(\Gamma)$ defined by 
\begin{equation}
Cf = f+2G^{\frac{1}{2}}_0\fint_{\Gamma^\ast}dq'\sum_{n=1}^N
\left(\frac{(\gamma_{\rm per}^0)^\perp_{q'}}{(H^0_{\rm per})_{q'}-\epsilon_{n,q'}}u_{n,q'}G_0^{\frac{1}{2}}f\right)
\overline{u_{n,q'}}
\label{formula_epsilon_tilde_P0}
\end{equation}
for all $f\in P_0L^2_{\rm per}(\Gamma)$.
\item One has for all $\sigma\in S^2$
\begin{equation}
\lim_{\eta\to0^+}\pscal{e_0,\tilde{\epsilon}^{-1}_{\eta\sigma}e_0}=\frac{1}{1+\sigma^TL\sigma-\pscal{P_0b_\sigma,C^{-1}P_0b_\sigma}}.
\label{eq:def_dielectric} 
\end{equation}
\end{enumerate}
\end{lemma}

Using \eqref{eq:def_dielectric}, we may now define the macroscopic dielectric permittivity as follows:
$$
\boxed{k^T\epsilon_{\rm M}k:=|k|^2+k^TLk-\pscal{P_0b_{k},C^{-1}P_0b_{k}}.}
$$
As $P_0b_k$ is linear in $k$, it follows that $\epsilon_{\rm M}$ is a
constant $3\times3$ symmetric matrix. 

Theorem \ref{thm:homogenization} readily follows from 
\begin{lemma}[Limit of the linear term]\label{lem:weak_limit} Let $g$ be a
  fixed function in $L^2(\R^3)$. Then
  $U_\eta^\ast\tilde{\epsilon}^{-1}U_\eta g$ weakly converges in
  $L^2(\R^3)$ as $\eta\to0^+$ to the function whose Fourier transform is
  given by 
$$\frac{\widehat{g}(k)}{1+\frac{k^TLk}{|k|^2}-\pscal{P_0b_{\frac{k}{|k|}},C^{-1}P_0b_{\frac{k}{|k|}}}}.$$ 
\end{lemma}

Assuming Lemma \ref{lem:def_epsilon}, we first write the 
\begin{proof}[Proof of Lemma \ref{lem:weak_limit}]
As $\tilde{\epsilon}^{-1}$ is a bounded operator on $L^2(\R^3)$, it suffices to show that
$$\lim_{\eta\to0^+}\pscal{U_\eta^\ast\tilde{\epsilon}^{-1}U_\eta g,g'} =  \int_{\R^3}\frac{\overline{\widehat{g}(k)}\widehat{g'}(k)}{1+\frac{k^TLk}{|k|^2}-\pscal{P_0b_{\frac{k}{|k|}},C^{-1}P_0b_{\frac{k}{|k|}}}}dk$$
for two functions $g,g'\in L^2(\R^3)$ such that both $\widehat{g}$ and $\widehat{g'}$ have a compact support (say in a ball of radius $R$). As the Fourier transforms of $U_\eta g$ and of $U_\eta g'$ have their support in the ball of radius $R\eta$, for $\eta$ small enough such that $B(0,R\eta)\subset \Gamma^\ast$, we have by the definition of the Bloch-Floquet transform
\begin{align*}
\pscal{\tilde{\epsilon}^{-1}U_\eta g,U_\eta g'}&= \int_{\R^3}\pscal{(\tilde{\epsilon}^{-1})_ke_0,e_0}\overline{\widehat{U_\eta g}(k)}\widehat{U_\eta g'}(k)\,dk\\
&= \int_{\R^3}\pscal{(\tilde{\epsilon}^{-1})_{\eta k}e_0,e_0}\overline{\widehat{g}(k)}\widehat{g'}(k)\,dk.
\end{align*}
The result then follows from \eqref{eq:def_dielectric} and the dominated
convergence Theorem. 
\end{proof}

It now remains to write the
\begin{proof}[Proof of Lemma \ref{lem:def_epsilon}]
Using \eqref{eq:FB_L} and time-reversal symmetry, we deduce that for all $f\in L^2_{\rm per}(\Gamma)$,
\begin{multline}
\left(\tilde\epsilon_qf\right)(x)=f+2\fint_{\Gamma^\ast}dq'\sum_{n=1}^N\sum_{n'=N+1}^{+\ii} \\
\frac1{\epsilon_{n',q'+q}-\epsilon_{n,q'}}\pscal{u_{n',q'+q},u_{n,q'}(v_c)^{\frac{1}{2}}_qf
 }_{L^2_{\rm per}(\Gamma)}\left[(v_c)^{\frac{1}{2}}_q\big(\overline{u_{n,q'}}u_{n',q'+q}\big)\right](x),
\label{formula_epsilon_tilde1}
\end{multline}
where $(v_c)^{\frac{1}{2}}_q$ is the convolution operator by the
corresponding Bloch component, which just consists in multiplying the
$K^{\rm th}$ Fourier coefficient of a function by
$(4\pi)^{\frac{1}{2}}|K+q|^{-1}$. 
The above formula can be rewritten as
\begin{equation}
\tilde\epsilon_qf =f+2(v_c)^{\frac{1}{2}}_q\fint_{\Gamma^\ast}dq'\sum_{n=1}^N
\left(\frac{(\gamma_{\rm per}^0)^\perp_{q'+q}}{(H^0_{\rm per})_{q'+q}-\epsilon_{n,q'}}u_{n,q'}{(v_c)^{\frac{1}{2}}_qf}\right)
\overline{u_{n,q'}}.
\label{formula_epsilon_tilde2}
\end{equation}
We note that for any $q\in\Gamma^\ast\setminus\{0\}$, $(v_c)^{\frac{1}{2}}_q$ is
a bounded (indeed compact) operator on $L^2_{\rm per}(\Gamma)$ and that $q
\mapsto (v_{\rm c})_q^{\frac{1}{2}}$ continuous from $\Gamma^\ast\setminus\{0\}$
to ${\mathcal L}(L^2_{\rm per}(\Gamma))$. The continuity of $\tilde{\epsilon}_q$ when $q$ stays away from 0 is therefore easy to verify.

For $f\in P_0L^2_{\rm per}(\Gamma)$, we have 
$$
(v_c)_q^{\frac{1}{2}}f=\sum_{K\in\cR^\ast\setminus\{0\}} 
 \frac{\sqrt{4\pi} \,
  \widehat{f}_K}{|q+K|} \, \frac{e^{iK\cdot x}}{|\Gamma|^{\frac 12}}:=
G_q^{\frac{1}{2}}f
$$
where $G_q^{\frac{1}{2}}$ is the operator on $L^2_{\rm per}(\Gamma)$ which multiplies the $K^{\rm th}$ Fourier coefficient of a function $f$ by $\sqrt{4\pi}|q+K|^{-1}$ except the coefficient corresponding to $K=0$ which is replaced by zero. Note that $G_q^{\frac{1}{2}}\to G_0^{\frac{1}{2}}$ in norm.
Formula \eqref{formula_epsilon_tilde2} then shows that
$C(q):=P_0\tilde{\epsilon}_qP_0|_{P_0L^2_{\rm per}(\Gamma)}$ is bounded and converges as $q\to0$ to the operator $C$ defined on $P_0L^2_{\rm per}(\Gamma)$ by \eqref{formula_epsilon_tilde_P0}. Obviously, $C\geq1$ on $P^0L^2_{\rm per}(\Gamma)$.

Next, we have for $q$ small enough
\begin{align*}
&(\tilde\epsilon_q-1)e_0\\
&=\frac{2\sqrt{4\pi}}{|\Gamma|^{\frac 12}}\fint_{\Gamma^\ast}dq'\sum_{n=1}^N\sum_{n'=N+1}^{+\ii} 
\frac{\pscal{u_{n',q'+q},u_{n,q'}
 }_{L^2_{\rm
     per}(\Gamma)}}{|q|(\epsilon_{n',q'+q}-\epsilon_{n,q'})}(v_c)^{\frac{1}{2}}_q\big(\overline{u_{n,q'}}u_{n',q'+q}\big) \\
&=\frac{\cB(q)}{|q|^2}e_0+\frac{2\sqrt{4\pi}}{|\Gamma|^{\frac 12}}\fint_{\Gamma^\ast}dq'\sum_{n=1}^N\sum_{n'=N+1}^{+\ii}
\frac{\pscal{u_{n',q'+q},u_{n,q'}
 }_{L^2_{\rm per}(\Gamma)}}{|q|(\epsilon_{n',q'+q}-\epsilon_{n,q'})}G^{\frac{1}{2}}_q\big(\overline{u_{n,q'}}u_{n',q'+q}\big)
\label{formula_epsilon_1K}
\end{align*}
where $\cB(q)$ was defined before in \eqref{def_cB_q}.
Now we use \eqref{pscal_u_n_q} and get
\begin{multline*}
\tilde\epsilon_{\eta\sigma}e_0
=\left(1+\frac{\cB(\eta\sigma)}{\eta^2}\right)e_0
-\frac{2i\sqrt{4\pi}}{|\Gamma|^{\frac 12}}\fint_{\Gamma^\ast}dq'\sum_{n=1}^N\sum_{n'=N+1}^{+\ii}\\
\frac{\pscal{u_{n',q'+\eta\sigma},(\sigma\cdot\nabla)u_{n,q'}}_{L^2_{\rm
     per}(\Gamma)}}{(\epsilon_{n',q'+\eta\sigma}-\epsilon_{n,q'})(\epsilon_{n',q'+\eta\sigma}-\epsilon_{n,q'}-\eta q'\cdot\sigma-\frac{\eta^2}2)}G^{\frac{1}{2}}_q
\big(\overline{u_{n,q'}}u_{n',q'+\eta\sigma}\big).
\end{multline*}
Hence as $\eta\to0^+$, $\tilde\epsilon_{\eta\sigma}e_0$ converges strongly in $L^2_{\rm per}(\Gamma)$ to $b_\sigma$.

The last step is to use the Schur complement formula which tells us that
$$\pscal{\tilde{\epsilon}_q^{-1}e_0,e_0}=\frac{1}{\pscal{\tilde{\epsilon}_qe_0,e_0}-\pscal{P_0\tilde{\epsilon}_qe_0,C(q)^{-1}P_0\tilde{\epsilon}_qe_0}}.$$
The above convergence properties yield
$$\lim_{\eta\to0^+}\pscal{\tilde{\epsilon}_{\eta\sigma}^{-1}e_0,e_0}=\frac{1}{1+\sigma^TL\sigma-\pscal{P_0b_\sigma,C^{-1}P_0b_\sigma}}$$ 
as was claimed.
\end{proof}

\section*{Acknowledgements}

This work was initiated while we were visiting the Institute
for Mathematics and its Applications (IMA) in Minneapolis. We warmly
thank the staff of the IMA for their hospitality. 
This work was partially supported by the ANR grants LN3M and
  ACCQUAREL.


\end{document}